\documentclass[letterpaper, 11pt]{article}

\usepackage{CJKutf8}

\usepackage[whole]{bxcjkjatype}

\usepackage{bookmark}
\usepackage{type1cm}
\usepackage{enumerate}
\usepackage{algorithm}
\usepackage{algorithmicx}
\usepackage[noend]{algpseudocode}
\usepackage[sort]{natbib}

\usepackage{url}
\usepackage{hyperref}
\usepackage{threeparttable}
\usepackage[margin=1in]{geometry}
\usepackage{amsmath,amssymb,amsfonts,setspace} 
\usepackage{amsthm} 
\usepackage{latexsym}
\usepackage{comment}
\usepackage{color}
\usepackage{xcolor}
\usepackage{bm}
\usepackage[normalem]{ulem}
\usepackage{mathtools}
\usepackage{stmaryrd} 

\usepackage[geometry]{ifsym}

\setcitestyle{numbers,square,comma}

\hypersetup{
setpagesize=false,
 bookmarksnumbered=true,%
 bookmarksopen=true,%
 colorlinks=true,%
 linkcolor=red,
 citecolor=blue,
}

\newtheoremstyle{mystyle}
    {}
    {}
    {\normalfont}
    {}
    {\bf}
    {}
    { }
    {}%

\theoremstyle{mystyle}
\newtheorem{theorem}{Theorem}
\newtheorem{lemma}{Lemma}
\newtheorem{proposition}{Proposition}
\newtheorem{corollary}{Corollary}
\newtheorem{definition}{Definition}
\newtheorem{example}{Example}

\newcommand{\Polylog}{\mathrm{polylog}}

\newcommand{\Gcal}{\mathcal{G}} %
\newcommand{\Tcal}{\mathcal{T}} %
\newcommand{\Pcal}{\mathcal{P}} %

\newcommand{\Ccolor}[1]{C_{\mathsf{col}(#1)}}
\newcommand{\Ccount}[1]{C_{\mathsf{cnt}(#1)}}

\newcommand{\TW}[1]{\tau} %
\newcommand{\CH}{\mathsf{ch}} %
\newcommand{\CHT}{\ensuremath{\mathsf{cht}}} %
\newcommand{\St}{c^{\ast}} %

\newcommand{\Comm}[1]{\llbracket{#1}\rrbracket} %

\newcommand{\Eout}[1]{E^{\mathrm{out}}_{#1}}

\newcommand{\Mem}{M}
\newcommand{\Trans}{\delta}

\newcommand{\Auxinit}{\triangledown}

\newcommand{\Par}[1]{p({#1})} %

\newcommand{\Dlabel}{\mathsf{la}} %
\newcommand{\Ddec}{\mathsf{dec}}  %
\newcommand{\Slabel}{\mathsf{sla}}%
\newcommand{\Sdec}{\mathsf{sdec}} %

\newcounter{cntLemmaNumber}
\newenvironment{rlemma}[1]{%
\setcounter{cntLemmaNumber}{\thelemma}
\setcounterref{lemma}{#1}
\addtocounter{lemma}{-1}
\begin{lemma}
}{%
\end{lemma}
\setcounter{lemma}{\thecntLemmaNumber}
}
\newcounter{cntTheoremNumber}
\newenvironment{rtheorem}[1]{%
\setcounter{cntTheoremNumber}{\thetheorem}
\setcounterref{theorem}{#1}
\addtocounter{theorem}{-1}
\begin{theorem}
}{%
\end{theorem}
\setcounter{theorem}{\thecntTheoremNumber}
}
\newcounter{cntPropositionNumber}

\allowdisplaybreaks[2]

\title{Fully Polynomial-Time Distributed Computation \\ in Low-Treewidth Graphs}

\author{
Taisuke Izumi\footnote{Graduate School of Information Science and Technology, Osaka University, 1-5, Yamadaoka, Suita, Osaka, 565-0871, Japan. E-mail: t-izumi@ist.osaka-u.ac.jp.}
\and Naoki Kitamura\footnote{Graduate School of Engineering, Nagoya Institute of Technology, Gokiso-cho, Showa-ku, Nagoya, Aichi, 466-8555, Japan. E-mail: ktmr522@yahoo.co.jp.}
\and Takamasa Naruse\footnote{Graduate School of Engineering, Nagoya Institute of Technology, Gokiso-cho, Showa-ku, Nagoya, Aichi, 466-8555, Japan. E-mail: t.naruse.333@stn.nitech.ac.jp.}
\and Gregory Schwartzman\footnote{Japan Advanced Institute of Science and Technology (JAIST), 
1-1 Asahidai, Nomi, Ishikawa 923-1292 Japan. E-mail: greg@jaist.ac.jp.}
}

\date{}

\begin{document}

\maketitle

\begin{abstract}
    We consider global problems, i.e. problems that take at least diameter time, even when the bandwidth is not restricted. We show that all problems considered admit efficient solutions in low-treewidth graphs. By ``efficient'' we mean that the running time has polynomial dependence on the treewidth, a linear dependence on the diameter (which is unavoidable), and only a \emph{polylogarithmic} dependence on $n$, the number of nodes in the graph. We present the following results in the CONGEST model (where $\TW{G}$ and $D$ denote the treewidth and diameter of the graph, respectively):
    \begin{itemize}
        \item Exact single-source shortest paths (Actually, the more general problem of computing a distance labeling scheme) for weighted and directed graphs can be computed in $\tilde{O}(\TW{G}^2D + \TW{G}^5)$ rounds\footnote{The $\tilde{O}(\cdot)$ notation hides $\mathrm{polylog}(n)$ factors.}. This is the first exact algorithm for the directed single-source shortest paths problem in low-treewidth graphs attaining a $\tilde{O}(\TW{G}^{O(1)}D)$-round running time.
        \item Exact bipartite unweighted maximum matching can be computed in $\tilde{O}(\TW{G}^{4}D + \TW{G}^7)$ rounds. This is the first algorithm for a non-trivial graph class that achieves a worst case running time \emph{sublinear} in the input size.
        \item The weighted girth can be computed in $\tilde{O}(\TW{G}^2D + \TW{G}^5)$ rounds for both directed and undirected graphs. Our results are the first to imply an \emph{exponential} separation between the complexity of computing girth and diameter for a non-trivial graph class.
    \end{itemize}
    Although the above problems are seemingly unrelated, we derive all of our results using a single unified framework. Our framework
    consists of two novel technical ingredients. The first is a
    fully polynomial-time distributed tree decomposition algorithm,
    which outputs a decomposition of width $O(\TW{G}^2\log n)$ in $\tilde{O}(\TW{G}^{2}D + \TW{G}^{3})$ rounds (where $n$ is the number of nodes in the graph). 
    The second ingredient, and the technical highlight of this paper, is the novel concept of a \emph{stateful walk constraint}, which naturally defines a set of feasible walks in the input graph based on their local properties (e.g., augmenting paths). Given a stateful walk constraint, the constrained version of the shortest paths problem (or distance labeling) requires the algorithm to output the shortest \emph{constrained} walk (or its distance) for a given source and sink vertices. We show that this problem can be efficiently solved in the CONGEST model by reducing it to an \emph{unconstrained} version of the problem. 
\end{abstract}

\newpage

\section{Introduction}

\subsection{Background and Motivation}
\label{subsec:background}

The \emph{treewidth} is one of the most important graph parameters and has received a huge amount of attention in the context of centralized algorithms \citep{Cygan15}. Informally speaking, it represents 
the graph's similarity to a tree (e.g., a tree has treewidth 1, a cycle has treewidth 2, a clique on $n$ nodes has treewidth $n-1$). In the context of centralized algorithms, a vast class of 
computationally hard problems is efficiently solvable in low-treewidth graphs. Furthermore, many real-world data sets are indeed low-treewidth graphs \cite{ManiuSJ19}.

In this work, we focus on the CONGEST model of distributed computation (see Section~\ref{sec: congest} for a formal definition). The inherent bandwidth limitation in the CONGEST model precludes any efficient centralized solution by aggregating the entire 
topological information of the network, and thus our algorithms must make do with only \emph{local} information. 
The usefulness of the treewidth parameter in the CONGEST is mostly due to the recent framework of \emph{low-congestion shortcuts}~\citep{GH16-2,HIZ16-2,HLZ18}, which provides efficient group communication for a collection of subgraphs \citep{GH16-2}. Based on this framework, several algorithms have achieved near-optimal running times for various fundamental problems in low-treewidth graphs~\citep{GH16-1,GH16-2,HIZ16-2,HLZ18}. For example, minimum-spanning tree, minimum-cut approximation, and approximate undirected single-source shortest paths~\citep{GH16-2,HIZ16-2,HL18,ZGYHS22}. 

While the low-congestion 
shortcut framework is a valuable tool for designing CONGEST algorithms (and is also used in this paper), 
it is a general framework, not limited to any specific graph class. Thus, it leaves many intriguing open questions for the family of low-treewidth graphs. For example, the problems of computing an \emph{efficient} tree decomposition and designing efficient algorithms for fundamental problems
based on tree decomposition make explicit use of the structure of low-treewidth graphs. Currently, the only relevant result is due to Li~\citep{Li18}, which presents a CONGEST algorithm with a running time of $\tilde{O}(\TW{G}^{\TW{G}}D)$ rounds that computes a tree decomposition of width $O(\TW{G})$ with applications to the distributed computation of optimal solutions for classic NP-hard problems (e.g., vertex cover) whose running time exponentially depends only on the width of the computed decomposition. 
\subsection{Our Results}
\label{subsec:contribution}
We focus on the study of fully polynomial-time distributed computation 
in low-treewidth graphs. Where ``fully polynomial-time'' means that the running time of 
algorithms depends polynomially on the treewidth of the input graph, linearly on its unweighted 
diameter $D$ and only has polylogarithmic dependence on $n$. This can be seen as a distributed analogue of the recent work of \citep{Fomin18}, which considers problems whose non-parametrized complexity has a super-linear dependence on the input size and presents algorithms whose running time depends polynomially on the treewidth and only \emph{linearly} on the input size.
As most of the problems which admit an improved running time for low-treewidth graphs are \emph{global} problems,
they admit the universal lower bound of $\Omega(D)$ rounds in the distributed setting. Where the term ``universal'' means 
that the lower bound holds for {\em any} instance. Hence our analog of linear dependence on the
input size as a linear dependence on $D$ is very natural. 

Our results are not the first to achieve a fully polynomial-time dependence on the treewidth. Specifically, the shortcut-based MST and approximate min-cut algorithms mentioned in Sec.~\ref{subsec:background} require $\tilde{O}(\TW{G}D)$ rounds, which beats the $\tilde{\Omega}(\sqrt{n} + D)$-round lower bound for general graphs~\cite{DHKKNPPW11}. However, due to the general nature of the shortcut framework, it does not take full advantage of the structure of specific graph classes. This is exemplified in the recent work of \citep{GP18,LP19,Parter20} on planar graphs.
While planar graphs admit efficient shortcut-based algorithms~\citep{GH16-1,GH16-2}, it is possible to achieve improved results and tackle new problems by leveraging techniques that are specific to planar graphs.
Our research can be seen as a low-treewidth counterpart 
of these results. The main contribution of this paper is a single algorithmic framework from which we are able to derive fully polynomial-time algorithms in low-treewidth graphs for a set of seemingly unrelated problems. In what follows, we explain the details of our results.

\paragraph{Distance Labeling and Single-Source Shortest Paths}
Distance labeling (\textsf{DL}) is the problem of assigning vertices with short labels
such that it is possible to compute the distance from $u$ to $v$ only by using their labels. 
The standard single-source shortest paths problem (\textsf{SSSP}) is easily reduced to distance labeling: 
the source node simply distributes its label to all other nodes. We present a randomized 
algorithm for exact \emph{directed} \textsf{DL}, which correctly constructs all of the labels in 
$\tilde{O}(\TW{G}^{2}D + \TW{G}^5)$ rounds with high probability (whp)\footnote{Throughout this paper, the term ``with high probability'' means that the probability is at least $1 - 1/n^c$ for an arbitrary large constant
$c > 0$.}.
It is known that the undirected weighted shortest-path problem require $\tilde{\Omega}(\sqrt{n} + D)$ rounds~\citep{DHKKNPPW11} for general graphs, even for approximate solutions. The first improvement of this bound for low-treewidth graphs is due to Haeupler and Li~\citep{HL18}. Their algorithm applies to any undirected graph that admits good shortcuts, including low-treewidth graphs (the running time and approximation factor depend on the quality of the shortcut). However, the approximation factor achieved is super-constant and their results do not extend to \emph{directed} graphs. Concurrently and independently of our work, the approximation ratio and running time of \citep{HL18} was recently improved to $(1 + \epsilon)$ and $\tilde{O}(\TW{G}D n^{o(1)})$~\cite{ZGYHS22}. However, their results do not apply for \emph{exact} distance computation nor to directed graphs. 

\paragraph{Exact Maximum Matching}
We present a randomized algorithm that computes exact unweighted maximum matching in bipartite 
graphs running in $\tilde{O}(\TW{G}^4D + \TW{G}^7)$ rounds.
While the maximum matching problem has received much attention
in the context of distributed approximation~\citep{FTR06,AKO18,LPP08,BCGS17}, the complexity 
of finding an exact solution is yet unknown. For general graphs, \citep{BKS19} were the first to present a non-trivial algorithm with a running time of $O(s^2_{\max})$ rounds, where $s_{\max}$ is the size of the maximum matching. This was recently improved to $O(s^{3/2}_{\max})$ \citep{KI21}.
For the case of bipartite graphs, an algorithm by Ahmadi 
et al.~\citep{AKO18} is the only result for exact maximum matching. It achieves a running time of $\tilde{O}(s_{\max})$ rounds (and thus the worst-case bound is $\tilde{O}(n)$, even in low-treewidth graphs). 
We present the first algorithm for a non-trivial graph class which achieves a running time sublinear in $n$.

\paragraph{Weighted Girth}
We present a randomized algorithm that computes the weighted girth, $g$, of a directed or undirected input
graph in $\tilde{O}(\TW{G}^{2}D + \TW{G}^5)$ rounds with high probability.
The best known upper bound for computing $g$ in general graphs is $\tilde{O}(\min\{gn^{1-\Theta(1/g)}, n\})$ rounds~\citep{CFGLLO20}, and a lower bound of $\tilde{\Omega}(\sqrt{n} + D)$ rounds is also known for unweighted
and undirected graphs~\citep{FHW12}. The lower bound holds even for a $(2 - \epsilon)$-approximation of $g$. Planar graphs admit an efficient solution 
for computing $g$ in $\tilde{O}(D^2)$ rounds \citep{LP19}. Our approach is structural and is 
based on a very simple (randomized) reduction of girth computation to a distance labeling scheme. Our techniques are novel and we believe that they may be applicable to other graph classes. 

Our result is the first separation between the complexity of diameter
computation and girth computation in undirected and unweighted graphs. All previously known results  exhibit similar complexity bounds for both problems, i.e., $\tilde{\Omega}(n^{\Theta(1)})$-round 
lower bounds for general graphs, and $\tilde{O}(D^{O(1)})$-round upper bounds for planar graphs~\citep{LP19,Parter20}. On the other hand, there exist hard instances of constant diameter 
and logarithmic treewidth for which computing the diameter in the CONGEST model requires $\Omega(n)$ rounds~\citep{AHK16}, contrasted with our algorithm for computing girth. That is, graphs of logarithmic treewidth and diameter are the first non-trivial graph class that exhibits an exponential separation in the round complexity for these two fundamental problems.

\subsection{Our Framework}
All of our algorithms are a direct result of a single unified framework.
The key technical ingredients of our framework are twofold: A new \emph{fully polynomial-time 
tree-decomposition algorithm}, 
and the novel concept of a \emph{stateful walk constraint}. In this subsection, we outline their ideas and their applications for our algorithms. 

\paragraph{Fully Polynomial-Time Tree Decompostion}
All of our results require the existence of an \emph{efficient} tree decomposition algorithm 
with a \emph{small width}. Unfortunately, the best known tree-decomposition algorithm \citep{Li18} has a running time that exponentially depends on the treewidth, and thus is too slow for our needs. Thus, we develop a fully polynomial-time CONGEST algorithm for tree decomposition, which runs in $\tilde{O}(\TW{G}^2D + \TW{G}^3)$ rounds and computes a tree decomposition 
of width $\tilde{O}(\TW{G}^2)$. The algorithmic ideas are based on the fully-polynomial 
time (centralized) 
tree decomposition algorithm by Fomin et al.~\citep{Fomin18}, with several nontrivial modifications that 
allow for an efficient implementation in the CONGEST model. While a direct implementation of~\citep{Fomin18} in the CONGEST model is straightforward, this will result in a round complexity of $\tilde{O}(\TW{G}^{O(1)}D)$, where the exponent of $\TW{G}$ is (at least) 7. We introduce novel ideas which allow us to substantially improve the dependence on $\TW{G}$. Our distance labeling result is obtained
by combining this tree decomposition algorithm with several techniques by Li and Parter~\cite{LP19} which were introduced in the context of distance labeling for planar graphs.

\paragraph{Stateful Walk Constraint}
The second ingredient of our framework is to extend the applicability of 
distributed directed shortest paths algorithms (including distance labeling schemes) to a more general type of shortest walks. 
We consider a constrained version of \textsf{SSSP} (or \textsf{DL}), where a subset $C$ of all walks 
in the input graph is given. This problem requires that each node $v$ in $G$ knows the length of the shortest 
walk in $C$ from a source vertex, $s$, to $v$ (or construct a labeling scheme that allows computing the length of the shortest walk in $C$ connecting the two vertices, using only their labels). 
This problem is not meaningful if $C$ is \emph{explicitly} given to each node, and our focus is 
the scenario where $C$ is given in an \emph{implicit} and \emph{distributed} manner. 

We introduce a natural 
class of walk constraints, which we call \emph{stateful walk constraints}. Roughly speaking, a stateful walk constraint is a set, $C$, of walks such that each node $u$ can \emph{locally} decide if a walk leaving $u$ is contained in $C$ or not, using only a small amount of additional information (referred to as 
the \emph{state} of a walk). This class captures many natural walks with combinatorial constraints, such as alternating walks (used in our matching algorithm). We show that the constrained versions of directed $\textsf{SSSP}$ and $\textsf{DL}$ under a stateful-walk constraint can be reduced to the corresponding unconstrained versions with a running-time overhead depending on the size of the state space associated with $C$. 

Let us outline how to apply the above framework to the problems of matching and girth. We show 
that each of these problems can be reduced to finding shortest walks under some
stateful-walk constraint. Our maximum matching algorithm, (i.e., alternating path finding) is one of the most natural applications of the framework. By combining the 
stateful-walk framework with a specific property of augmenting paths in low-treewidth graphs (presented in ~\citep{IOO18}), we derive our algorithm. For girth computation, our key idea is to use the framework to exclude walks that ``fold onto themselves", that is, the second half of the walk is the inversion of the first half. This leaves us with a set of walks which \emph{upper bound} the girth. Finally, we use a probabilistic sampling of edge labels combined with the above to derive our algorithm. 

The above applications demonstrate the expressive power and versatility of our framework. Finally, we would like to emphasize that the framework is not limited to low-treewidth graphs, but applies to general graphs. The authors believe that this framework is potentially useful in the design of 
efficient CONGEST algorithms for a wider class of problems.

\subsection{Related Work} \label{sec:relatedwork}
Distance computation problems are at the core of 
distributed graph algorithms. Recently there has been a vast number of results for both exact 
and approximate distance computation problems \citep{LP13,LP13-2,FHW12,HW12,DDP12,PRT12,Nanongkai14,HKN16,GL18,HNS17,FN18,BN18,BKKL17,Elkin17,AHK16,IW14,ARKP18,HL18,GKP20,CFGLLO20,CM20}. The state-of-the-art bounds
for general graphs are $\tilde{O}(\sqrt{n} + D)$ rounds for 
$(1 + \epsilon)$-approximate \textsf{SSSP}~\citep{BKKL17}, and 
$\tilde{O}(\sqrt{n}D^{1/4} + D)$ rounds for exact \textsf{SSSP}~\citep{CM20}. These results
hold for weighted and directed graphs. A tight runtime bound for exact \textsf{SSSP} is still an open
problem. For computing the girth, a near optimal approximation
algorithm for unweighted girth is known~\citep{PRT12}. It outputs the girth with an additive 
error of one (i.e., $(2 - 1/\Theta(g))$-multiplicative approximation) in $O(\sqrt{ng}\log n + D)$ 
rounds. 

Planar graphs are also an intriguing class of graphs, and are closely related to our results.
Although our algorithms are applicable to planar graphs, as planar graphs have treewidth of $O(D)$, 
for distance computation problems the existing algorithms tailored for planar graphs~\citep{LP19,Parter20} achieve a running time with a better dependence on $D$.

Distance labeling schemes were first proposed by Gavoille et al.~\citep{GPPR2004}, and their centralized construction was studied for many graph classes~\citep{GP03,FGNW17,GP08,KKP05,GPPR2004}. Distance labeling is closely related to 
\emph{(approximate) distance oracles}, which are centralized data structures for representing distance matrices that support quick access. There are a few results that consider distributed variants of approximate distance oracles~\citep{LP13,DDP12,IW14,LP15} for general graphs. However, all of them consider only approximate oracles, and essentially require the construction time to polynomially depend on $n$. The best bound for bounded treewidth graphs is the shortcut-based construction~\citep{HLZ18}, whose approximation factor is $\Polylog(n)$.

While most previous results for maximum matching focus on approximate 
solutions~\citep{FTR06,AKO18,LPP08,BCGS17}, the problem of finding an \emph{exact} solution has been receiving more and more attention recently~\citep{BKS19,AK20,KI21}. On the negative side, the lower bound of $\tilde{\Omega}(\sqrt{n} + D)$ rounds is implied from the result of~\citep{AKO18}. This lower bound is recognized as a strong barrier: It has been shown that the well-known approach of reduction from two-party communication complexity does not work for obtaining any stronger lower bound~\citep{BCDELP19}.
It is also known that exact maximum matching does not have any local solution: There exists a hard instance where $D = \Theta(n)$ that exhibits an $\Omega(n)$-round lower bound~\citep{BKS19}. 

\subsection{Organization of the Paper}
In Section~\ref{sec:preliminaries} we introduce the concept of tree 
decomposition and the CONGEST model. In Section~\ref{sec:TreeDecomposition}, we present our tree 
decomposition algorithm in the CONGEST model. In Section~\ref{sec:distancelabeling} we show how to solve the distance labeling problem by using our tree decomposition algorithm. In Section~\ref{sec:statefulwalk} we introduce the concept of stateful walks, and show how to reduce the problem of finding (shortest) stateful walks into the standard directed reachability or the shortest paths problem. Sections~\ref{sec:matching} and \ref{sec:girth} are devoted to the applications of our framework.

\section{Preliminaries}
\label{sec:preliminaries}

\subsection{Model and Notations}
\label{sec: congest}
Let us now define the CONGEST model of distributed computation. We model distributed systems as an 
undirected and unweighted graph, $G$, on $n$ nodes, where the nodes are computational units and edges are communication links. We assume that nodes have unique $O(\log n)$ bit IDs. Communication between nodes happens in \emph{synchronous} rounds. In each round, each node sends a (possibly different) $O(\log n)$-bit message to each neighbor and, within the same round, receives all messages from the neighbors. After receiving the messages, it performs some local computation. We assume that nodes have unbounded computational power, and when analyzing our algorithms, we only care about the \emph{communication cost} of the algorithm. That is, the number of communication rounds it takes to complete. For any graph $G$, we denote its vertex and edge sets by $V(G)$ and $E(G)$ respectively.

While we also deal with directed and weighted multigraphs as input instances, the communication network itself is modeled as a simple undirected unweighted graph (i.e., the orientation, weight, and multiplicity of the edges connecting two vertices do not affect the communication capability between them). More precisely, given an input instance $G$, we denote by $\Comm{G}$ the graph which is obtained by omitting all orientations of $E(G)$, by merging the multiedges connecting the same two vertices into a single one, and by removing all self-loops. Then $\Comm{G}$ is the communication network implied by $G$. Given a graph $H$, we denote by $D(H)$ the undirected diameter of $H$ (i.e., $D(H)$ is the diameter of $\Comm{H}$). For the input graph $G$, we use $D$ instead of $D(G)$. For any rooted tree $T$ and a vertex $v \in V(T)$, we denote by $T(v)$ the subtree of $T$ rooted by $v$. We also denote
by $\CH(T, v)$ the set of the children of $v$ in tree $T$.

\subsection{Tree Decomposition and Treewidth}
\label{subsec:treedecomposition}
 Let $G = (V(G), E(G))$ be an undirected
and unweighted graph. 
A tree decomposition of an undirected and unweighted graph $G$ is 
a pair $\Phi = (T, \{B_x\}_{x \in V(T)})$, where $T$ 
is a tree, referred to as \emph{decomposition tree}, and each vertex $x \in V(T)$ is associated 
with a subset $B_x \subseteq V(G)$ of vertices in $G$ (called \emph{bag $x$}) satisfying the 
following conditions:
\begin{enumerate}[(a)]
	\item $V(G) = \bigcup_{x \in V(T)} B_x$.
	\item Any edge in $G$ is covered by at least one bag, i.e., for all $(u, v) \in E(G)$, there exists 
	$x \in V(T)$ such that $u, v \in B_x$ holds.
	\item For any $u \in V(G)$, the subgraph of $T$ induced by the bags containing $u$ is connected. 
\end{enumerate}
The \emph{width} of a tree decomposition $\Phi = (T, \{B_{x}\}_{x \in V(T)})$ is defined as the maximum bag size minus one. The \emph{treewidth}, $\TW{G}$, of a graph $G$ is the minimum width over all tree decompositions of $G$. 
While the original definition of treewidth applies only to undirected graphs, we define the treewidth
of a directed graph $G$ as the treewidth of $\Comm{G}$.

Throughout this paper, we assume that any decomposition tree $T$ is rooted, and each vertex in $V(T)$ is identified by a string over the alphabet $[0, n - 1]$. Letting $x$ be any string over the alphabet 
$[0, n - 1]$ and $i$ be any character,
we define $x \bullet i$ as the string obtained by adding $i$ to the tail of $x$. The null string of 
length zero is denoted by $\psi$, which is the identifier of the root of $T$. 
Given a vertex $x \in  V(T)$,
$x \bullet i$ identifies the $i$-th child of $x$. We use the notation $x \sqsubseteq y$ if $x$ is a prefix of $y$, and the notation $x \parallel y$ if neither $x \sqsubseteq y$ nor $y \sqsubseteq x$ holds.
We denote the length of $x$ by $|x|$, which means the depth of vertex $x$ in $T$. We define $A_\ell(T)$
as the set of vertices of length $\ell$ in $V(T)$.
For any tree decomposition $\Phi = (T, \{B_x\}_{x \in V(T)})$ 
and $v \in V(G)$, its \emph{canonical string} $\St_{\Phi}(v)$ is the shortest string such that $v \in B_{\St_{\Phi}(v)}$ holds. Note that $\St_{\Phi}(v)$ is uniquely determined
because the set of bags containing $v$ forms a connected subgraph of $T$ (by condition (c) of the definition).
The subscript $\Phi$ is often omitted when it is clear from context.
Letting $x$ be a string of non-zero length in $V(T)$, we denote by $\Par{x}$ the string corresponding 
to the parent of $x$ (i.e., the string obtained by chopping the tail of $x$) in $T$. We define 
$\CHT_{\Phi}(x)$ as the set of $i \in [0, n - 1]$ such that $x \bullet i$ is a child of $x$, 
i.e., $\CHT_{\Phi}(x) = \{i \mid x \bullet i \in V(T) \}$. 

In the distributed setting, computing a tree decomposition means that each node $u \in V(G)$ 
outputs the IDs of the bags containing $u$. 

\subsection{Part-wise Aggregation}
\label{subsec:partwiseaggregation}

Throughout this paper, we often execute an algorithm, $\mathcal{A}$, on multiple subgraphs of 
the input graph independently and simultaneously. That is, given a collection 
$\mathcal{H} = \{H_1, H_2, \dots, H_N\}$ of vertex disjoint connected subgraphs of 
the input graph $G$, we execute $\mathcal{A}$ on all $H_i \in \mathcal{H}$ in parallel.
The primary obstacle in implementing this type of execution in the CONGEST model is 
that the diameter $D(H_i)$ for $H_i \in \mathcal{H}$ may be much larger than 
$D(G)$ (and can be $\Omega(n)$ in the worst case), and thus the running time of $\mathcal{A}$ in $H_i$ 
can depend on $n$ even if the running time of $\mathcal{A}$ depends only on the diameter of 
the input graph. 
The key technical ingredient for this section is a subroutine called 
\emph{part-wise aggregation}~\citep{GH16-2}, which
is defined as follows:
Let $G = (V(G), E(G))$ be an undirected graph, $\mathcal{H} = \{H_1, H_2, \dots, H_N\}$ be a collection of connected vertex disjoint subgraphs of $G$, and $\oplus$ be an associative binary function 
operating on a value domain $\mathcal{M}$ of cardinality $\mathrm{poly}(n)$. Suppose that each node $v \in V(H_i)$ knows all the edges in $E(H_i)$ incident to $v$, and has a value $x_{v, i} \in \mathcal{M}$. Every node in $H_i$ wants to 
learn the value $\bigoplus_{v \in V(H_i)} x_{v, i}$, i.e., the aggregation with operator $\oplus$ over all of the values $x_{v, i}$ for $v \in V(H_i)$. It is known that bounded treewidth graphs admit a fast algorithm for part-wise aggregation~\citep{HIZ16-2,HHEW18,Li18}, which runs in $\tilde{O}(\TW{G}D)$ rounds.

\section{Fully Polynomial-Time Distributed Tree Decomposition}
\label{sec:TreeDecomposition}

\subsection{Balanced Separator}
\label{subsec:balancedsepartor}

Our tree decomposition algorithm is based on the computation of \emph{balanced separators}, 
which is a common technique used in many (centralized or distributed) tree decomposition algorithms. 
We first introduce the notion of a $(X, \alpha)$-balanced separator, which is a slight generalization
of a conventional balanced separator.

Let $X$ be any subset of $V(G)$. For a given vertex subset $Y \subseteq V(G)$, we define $\mu_X(Y) = |Y \cap X|$. 
We also use a similar notation $\mu_X(H)$ for any subgraph $H \subseteq G$ to mean $\mu_X(V(H))$. 
The subscript $X$ is omitted if it is clear from the context. 
An $(X, \alpha)$-balanced separator, $S$, of an undirected graph $G$ is 
a vertex subset whose removal divides $G$ into $N$ connected components $G_1, G_2, \dots, G_N$ such that 
$\mu_X(G_i)/\mu_X(G) \leq \alpha$ 
holds for any $1 \leq i \leq N$. A \emph{$(V(G), \alpha)$-balanced separator}
is simply called an $\alpha$-balanced separator of $G$.

It is well-known that any graph $G$ admits a $(1/2)$-balanced separator whose size is $\TW{G} + 1$ (Lemma 7.19 of \citep{Cygan15}), and that one can obtain a tree-decomposition algorithm of width 
$O(t \log n)$ from any balanced separator algorithm which outputs a separator $S$ of size $t$.
The best known CONGEST algorithm for finding an $\alpha$-balanced separator of size $O(\TW{G})$ 
(for constant $\alpha < 1$)~\citep{Li18} has a running time that exponentially depends on 
$\TW{G}$, and thus does not fit our goals. To get rid of the exponential dependency, we present 
a new CONGEST algorithm for computing balanced separators building on the ideas of the centralized algorithm by 
Fomin et al.~\citep{Fomin18} (referred to as \textsc{Flpsw} hereafter).

\subsection{The Algorithm by Fomin et al.}
\label{subsec:Fomin}

We first present the outline of \textsc{Flpsw}. We assume $X = V(G)$ for simplicity, 
but the algorithm can handle an arbitrary $X$. 
\textsc{Flpsw} runs with a parameter $t$, and is guaranteed to output
an $\alpha$-balanced separator of size $O(t^2)$ for $\alpha = 1 - \Theta(1) > 0$ when $\TW{G} + 1 \leq t$. 
In the case when $\TW{G}$ is unknown, 
one can combine \textsc{Flpsw} with a standard doubling estimation technique for $t$. 
To explain the algorithm we first introduce the notion of an $U_1$-$U_2$ vertex cut for $U_1, U_1 \subseteq V(G)$ 
as a generalization of the standard $s$-$t$ vertex cut, which is defined as a vertex
subset $Z \subseteq V(G) \setminus (U_1 \cup U_2)$ such that $U_1$ and $U_2$ belong to different connected components in $G - Z$. If $U_1$ intersects with $U_2$ or some edge crosses
between $U_1$ and $U_2$, the size of the $U_1$-$U_2$ vertex cut is defined as $\infty$. 

Let $S$ be a $(1/2)$-balanced separator of $G$ of size at most $t$ (recall that it necessarily exists). 
The algorithm first constructs any rooted spanning tree, $T$, of 
$G$, and decomposes it into a set $\Tcal$ of $\Theta(t)$ subtrees of size $\Theta(n/t)$ such that 
only their root vertices are shared among two or more subtrees in $\Tcal$. In what follows,
we refer to this type of decomposition as the \emph{splitting} of $T$, and to each subtree as a 
\emph{split tree}. Let $R$ be the set of the root vertices of split trees. We assume 
that the hidden constant in the cardinality of $\Tcal$ is sufficiently large, e.g., 
$\Tcal \geq 100t$. There are two cases to consider.
\begin{itemize}
\item (Case 1) $R$ does not intersect $S$: Since all trees in $\Tcal$ are vertex disjoint except for $R$,
at most $|S| = t$ trees in $\Tcal$ intersect $S$. Then there exist two split trees $T_1, T_2 \in \Tcal$ such that they belong to different connected components in $G - S$, i.e., the minimum 
$V(T_1)$-$V(T_2)$ 
vertex cut $Z$ has size at most $|S| = t$. 
The algorithm finds such a pair by computing the minimum $V(T_1)$-$V(T_2)$ vertex cut for all pairs 
$(T_1, T_2) \in \Tcal^2$. Once the pair is found, the algorithm outputs $Z$ as the separator. Note that
$Z$ is a $(X, 1 - \Theta(1/t))$-balanced separator of $G$ because both $V(T_1)$ and $V(T_2)$ contain 
$\Theta(n/t)$ vertices. If the algorithm fails to find such a pair, it concludes that the first case does not apply, and proceeds to case 2.
\item (Case 2) $R$ intersects $S$: The algorithm
simply removes $R$ from $G$ and outputs it as the separator. The removal of $R$ 
results in the deletion of at least one vertex in $S$ from $G$. 
\end{itemize}
\textsc{Flpsw} iterates the procedure above $2t$ times for the largest connected components 
of the remaining graph. After all iterations are complete, we obtain a $(1 - \Theta(1))$-balanced separator. That is, if the first case succeeds $t$ times, then $\Theta(1)$ fraction of vertices are separated. Otherwise, all vertices in $S$ are removed. Since each iteration adds $O(t)$ vertices to the output set,
the total size of the output separator is $O(t^2)$. As stated in the introduction, it is relatively straightforward to implement \textsc{Flpsw} in $\tilde{O}(\TW{G}^{O(1)}D)$ rounds if we do not care about optimizing the exponent of $\TW{G}$. 

\subsection{Our Algorithm}

We present a modified version of \textsc{Flpsw} 
which admits a more efficient CONGEST implementation (the implementation details are 
deferred to Appendix~\ref{appendix:treewidth}). The key differences between our algorithm and 
\textsc{Flpsw} are threefold. First, instead of solving 
the minimum $V(T_1)$-$V(T_2)$ vertex cut problem for all $(T_1, T_2) \in \Tcal^2$, we simply 
adopt a random sampling strategy for identifying a pair $(T_1, T_2)$ which has a cut of $O(t)$ vertices.
When $S \cap R = \emptyset$ holds, this strategy is guaranteed to succeed with a constant probability. 
Since one pair $(T_1, T_2)$ is sampled per iteration, it suffices to solve $O(t)$ instances of 
the minimum vertex cut in total.

The second idea is a parallelization-friendly algorithm for tree splitting. More precisely, the algorithm
manages a set of disjoint trees $\Tcal$, where initially $\Tcal = \{T\}$, and iteratively splits trees of
large size in $\Tcal$, and then adds back the resulting split trees to $\Tcal$ if they are
still large. This strategy admits an efficient CONGEST implementation because the splitting
of two different trees in $\Tcal$ can be performed in parallel. 

The third idea is to compute
$O(t)$ instances of the minimum vertex cut simultaneously at the final step (\textsc{Flpsw} performs this computation sequentially). Utilizing a careful scheduling technique,
we can execute $t$ independent instances of the minimum vertex cut problem in $\tilde{O}(t \TW{G} D + t^2 \TW{G})$
rounds, which is more efficient than sequential processing (which takes $\tilde{O}(t^2 \TW{G}D)$ rounds).

We present the centralized version of our algorithm (referred to as \textsc{Sep} hereafter). 
It works as follows:
\begin{enumerate}
\item  If $\mu(G) \leq 200t^2$, the algorithm outputs $X$ as a 
$(X, 14399/14400)$-balanced separator and halts. 
\end{enumerate}
When the algorithm does not halt at step 1, the algorithm iteratively applies 
the following steps (2 and 3) for $\hat{t} = \lceil 301t/300 \rceil$ times to the graphs $G_1, G_2, \dots, G_{\hat{t}}$, where $G_1 = G$ and 
the rest of the sequence is generated within the following steps.
\begin{enumerate}
\setcounter{enumi}{1}
\item At the beginning of the $i$-th iteration (for $G_i$), the algorithm constructs 
some spanning tree $T^{\ast}$ of $G_i$, and then split $T^{\ast}$ into several trees. 
In the $i$-th iteration, this splitting procedure, which we refer to as \textsc{Split}, maintains the two sets of 
trees $\Tcal$ and $\Tcal_i$, which initially store $\Tcal = \{T^{\ast}\}$ and $\Tcal_i = \emptyset$.
By a single invocation of \textsc{Split}, every tree $T \in \Tcal$ is split into a set of trees of size 
at least $\mu(G)/(12t)$ and at most $5\mu(T)/6$. The original 
tree $T$ is removed from $\Tcal$ after splitting. Each split tree is added to $\Tcal$ if its size is 
more than $\mu(G)/(4t)$, or to $\Tcal_i$ otherwise.
The splitting process terminates when $\Tcal$ becomes empty.

The details of the procedure \textsc{Split} are as follows: For any $T \in \Tcal_i$,
the algorithm finds the center vertex $c \in V(T)$ of $T$, i.e., the vertex such that removing it decomposes $T$
into several subtrees of size at most $\mu(T)/2$.
Now we regard $c$ as the root of $T$. Next, $\textsc{Split}$ removes all subtrees $T(v)$ for $v \in \CH(T, c)$ 
such that $\mu(T(v)) \geq \mu(G)/(12t)$ as split trees.

Let $T'$
be the remaining tree. If $\mu(T') < \mu(G)/(12t)$, we pick any tree $T(v)$ split in the first step, 
and merge $T'$ into $T(v)$ (Fig.~\ref{fig:splittree}(a)). The size of $T' + T(v)$ is bounded by $\mu(G)/(12t) + \mu(T)/2 
\leq \mu(T)/3 + \mu(T)/2 \leq 5\mu(T)/6$ (recall that any $T \in \mathcal{T}$ has a size at least $\mu(G)/(4t)$ 
and thus $\mu(G)/(12t) \leq \mu(T)/3$ holds). Otherwise, we further split $T'$ into several 
subtrees sharing $c$ as their roots. Let us fix some ordering of the children of $c$ in $T'$, denoted by 
$y_0, y_1, \dots, y_{\ell - 1}$, we define $Y_{a,b} = \bigcup_{a \leq h < b} V(T'(y_h))$. The algorithm computes the indices $0 = q_0, q_1, \dots, q_{\ell'} = 
\ell - 1$ such that $\mu(G)/(12t) \leq \mu(Y_{q_{h-1}, q_{h}}) < \mu(G)/(6t)$ holds for 
all $1 \leq h \leq \ell' - 1$ and $\mu(G)/(12t) \leq \mu(Y_{q_{h-1}, q_{h}}) < \mu(G)/(4t)$ holds for $h = \ell'$. Then we split $T'$ into $\ell'$ connected subtrees induced by 
$Y_{q_h, q_{h+1}} \cup \{c\}$. Since the subtree $T'(y)$ for any $y \in \CH(T, c)$ has a size
less than $\mu(G)/12t$, one can always obtain such a splitting. Each induced subtree is added to $\mathcal{T}_i$ because its size 
is necessarily at most $\mu(G)/(4t)$ (Fig.~\ref{fig:splittree}(b)).

It is easy to see that $\Tcal$ becomes empty after $O(\log t)$ invocations of \textsc{Split}. At which point 
$\Tcal_i$ is a set of split trees covering
$T^{\ast}$, whose size is in the range $[\mu(G)/(12t), \mu(G)/(4t)]$. 

\item We denote the set of root vertices of subtrees in $\Tcal_i$ by $R_i$. 
If $R^{\ast}_i = \bigcup_{1 \leq j \leq i} R_i$ is a $(X, 14399/14400)$-separator of $G$, the algorithm outputs 
it and halts. Otherwise, we define $G_{i+1}$ as the heaviest connected component of $G_i - R_i$ with 
respect to $\mu$.
\end{enumerate}
If the algorithm completes $\hat{t}$ iterations of the steps above without halting, the following 
step is performed.
\begin{enumerate}
\setcounter{enumi}{3}
\item For each $i \in [1, \hat{t}]$, the algorithm chooses 95 ordered pairs uniformly at random from 
$\Tcal_i \times \Tcal_i$, and compute the $V(T_1)-V(T_2)$ vertex cut for all chosen pairs.
If the computed cut size is at most $t$, the cut vertices are added to the set $Z$.
Finally, $Z$ is outputted if it is a $(X, 14399/14400)$-balanced separator. Otherwise \textsc{Sep} fails.
When \textsc{Sep} fails for $5\log n$ trials, it concludes that $\TW{G} + 1 > t$ (and runs 
again after doubling $t$). 
\end{enumerate}

\begin{figure}[ht]
	\begin{center}
		\includegraphics[clip, scale=0.4]{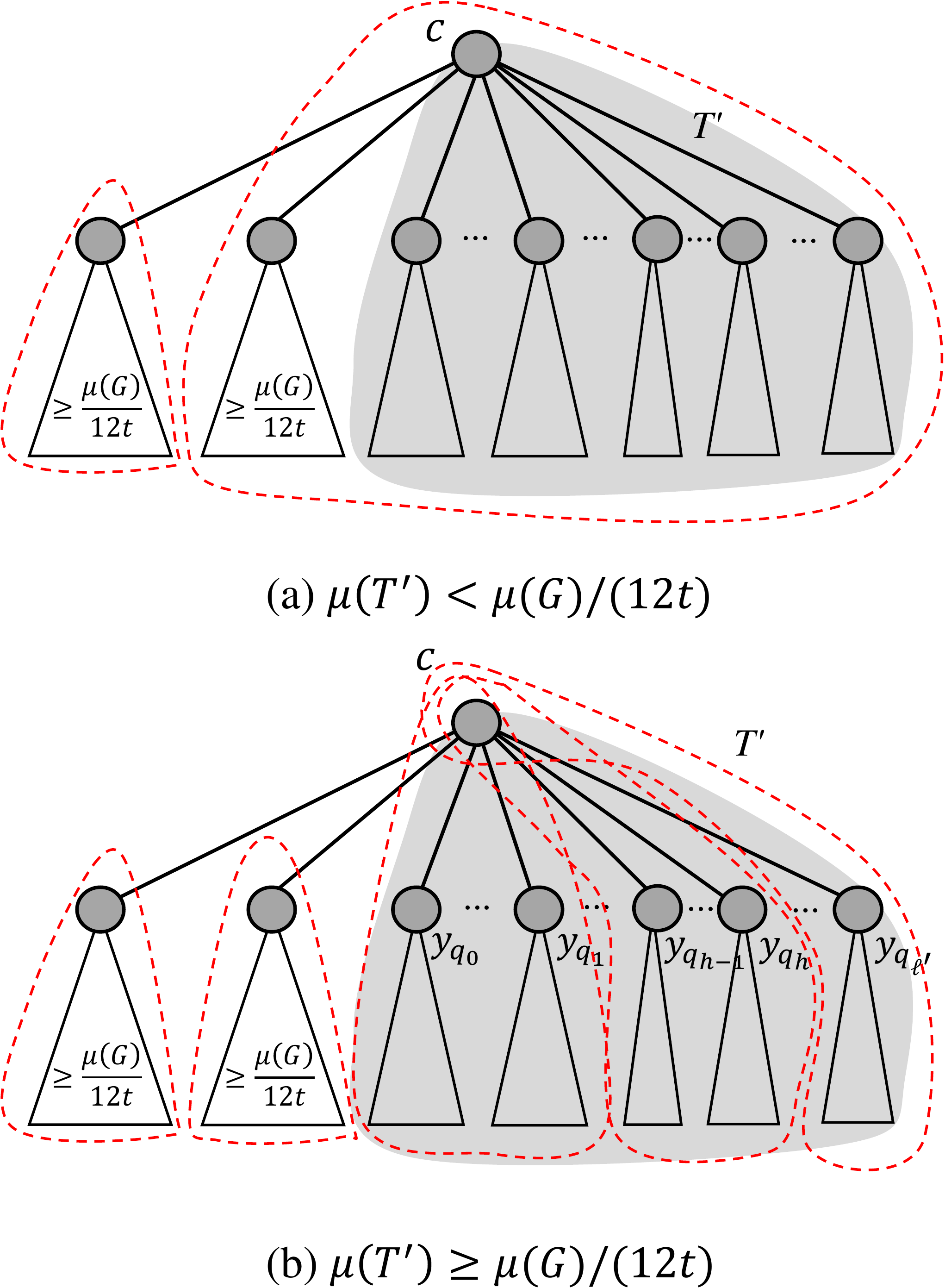}
	\end{center}
	\caption{\label{fig:splittree} An illustration of the \textsc{Split} procedure. Subtrees circled by 
	a red dotted line are split subtrees.} 
\end{figure}

While the fundamental idea of \textsc{Sep} is similar to \textsc{Flpsw}, it 
requires a completely new analysis and correctness proof. One significant technical challenge that we overcome is that the subtree pairs chosen in Step (4) are not vertex disjoint. This is problematic as we must prove that the size of the largest connected component after removing 
the computed separator $Z$ becomes substantially smaller. The complete proof of correctness
is deferred to Appendix~\ref{appendix:correctnesssep}. 
To implement \textsc{Sep} efficiently 
in the CONGEST model, we utilize the part-wise aggregation technique explained in 
Section~\ref{subsec:partwiseaggregation}, which is also known to provide efficient algorithms 
for minimum $U_1$-$U_2$ vertex cut and spanning tree construction running in 
$\tilde{O}(\TW{G}^{O(1)}D)$ rounds~\citep{GH16-2,HIZ16-2,Li18}. The details of the implementation and
its correctness proof are presented in Appendix~\ref{appendix:distsep}. Finally, we obtain 
the following lemma:

\begin{lemma}
    \label{lma:balancedseparatorconstruction}
    Let $G$ be an undirected graph, and $X \subseteq V(G)$ be any vertex subset. There exists 
    a randomized CONGEST algorithm that outputs a $(X, 14399/14400)$-balanced 
    separator of size at most $400(\TW{G}+1)^2$ for $G$ in $\tilde{O}(\TW{G}^2D + \TW{G}^{3})$ 
    rounds whp. 
\end{lemma}

\subsection{Distributed Tree Decomposition based on Balanced Separators}
\label{subsec:disttreedecomosition}

We construct a tree decomposition of width at most $O(\TW{G}^2 \log n)$, utilizing the balanced 
separator algorithm of lemma~\ref{lma:balancedseparatorconstruction}. We refer to the constructed 
tree decomposition as $\Phi = (T, \{B_x\}_{x \in V(\Tcal)})$.
As explained in Section~\ref{subsec:treedecomposition}, 
the subscript $x$ of each bag is a string over the alphabet $[0, n-1]$. Initially, let $G_{\psi} = G$.
There exists a standard strategy to obtain a decomposition from any balanced separator algorithm, 
which works as follows: We first compute a balanced separator $S$ of $G = G_\psi$. The set $S$ 
becomes the root bag $B_\psi$ of the constructed tree decomposition. For each connected component 
$G_0, G_1, \dots, G_{N-1}$ of $G - S$, we recursively construct their tree decompositions. Finally, 
we add $S$ to all of the bags in those decompositions, and connect their roots $1, 2, \dots, N$ to the root $\psi$ of the whole tree decomposition as children. Using an algorithm for computing a balanced separator of 
size $O(\TW{G}^2)$, this strategy yields a tree decomposition of size $O(\TW{G}^2\log n)$. However, adopting 
this strategy to the distributed setting is problematic, mainly due to the fact that the bag $B_x$ is 
not a subset of the vertices of the corresponding graph $G_x$. To avoid it, our algorithm utilizes a slightly modified strategy.

The algorithm recursively decomposes $G_{x}$ for each string $x$ by fixing the corresponding bag 
$B_x$. It first computes a $O(1)$-balanced separator $S_x$ of $G_x$ using 
the algorithm of Lemma~\ref{lma:balancedseparatorconstruction}. If $|V(G_{x})| \leq 2|S_x|$, 
we define $B_x = |V(G_{x})|$ and the recursion terminates. Otherwise, 
we define the bag $B_x = V(G_x) \cap (\bigcup_{x' \sqsubseteq x} S_{x'})$.
Let $G'_{x \bullet 0}, G'_{x \bullet 1}, \allowbreak \dots, G'_{x \bullet (N - 1)}$ be the connected components of 
$G_{x} - B_x$. The graph $G_{x \bullet i}$ (for $0 \leq i  \leq N-1$) is defined as $G_{x \bullet i} = 
G'_{x \bullet i} + \{(u, v) \mid (u, v) \in (V(G'_{x \bullet i}) \times B_x) \cap E(G_{x})\}$. 
This decomposition strategy is illustrated in Fig.~\ref{fig:treedecomposition}
This construction yields a tree decomposition of width $O(\TW{G}^2 \log n)$ and guarantees 
$B_x \subseteq G_x$ for any $x \in V(T)$.

\begin{figure}[ht]
	\begin{center}
		\includegraphics[clip, scale=0.3]{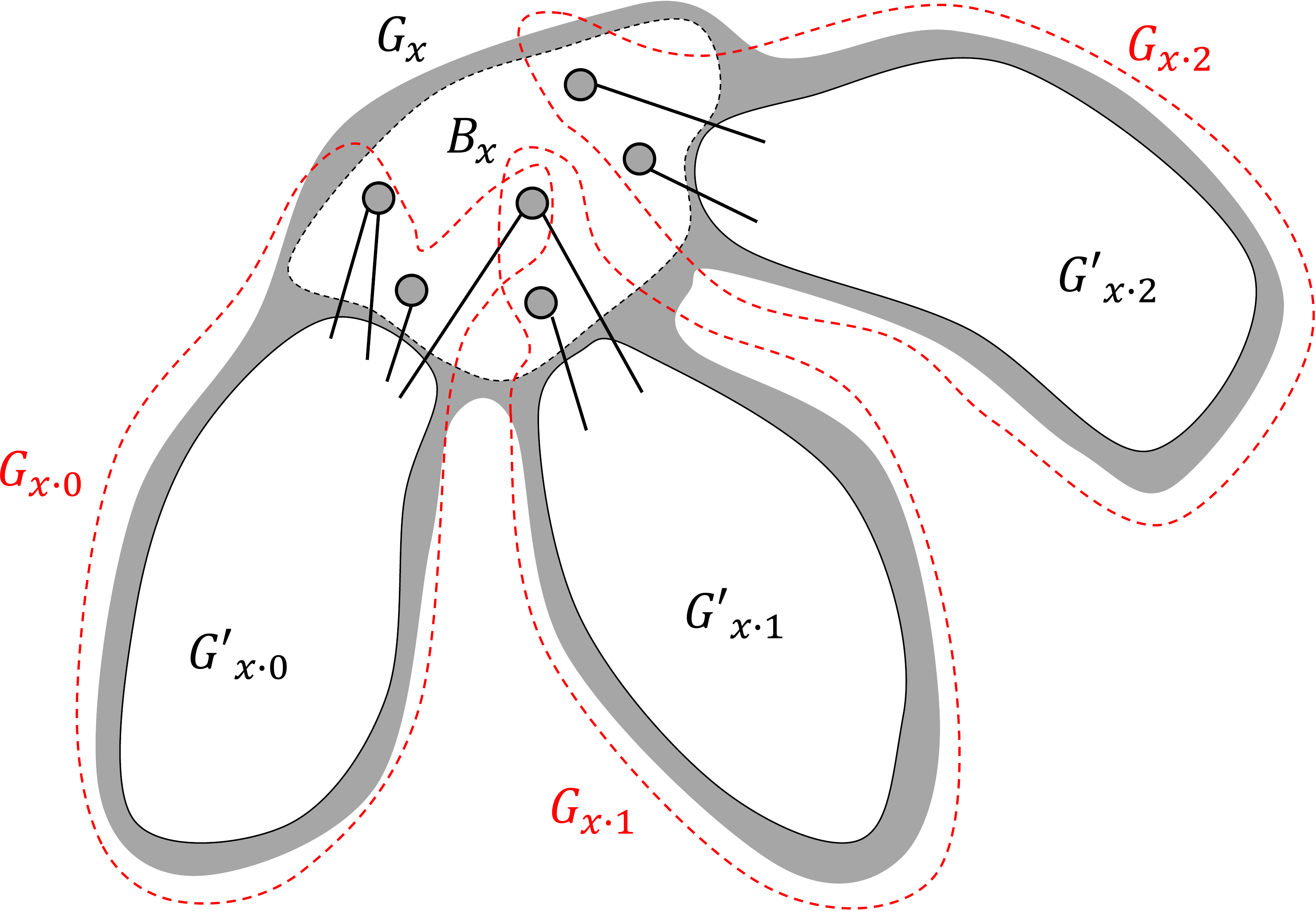}
	\end{center}
	\caption{\label{fig:treedecomposition} An illustration of our tree decomposition.}
\end{figure}

A primary obstacle in adopting this decomposition strategy to the distributed setting is the fact
that the collection of subgraphs $\mathcal{G}_\ell = \{G_x \mid x \in A_\ell\}$ for any $\ell$ is
not necessarily vertex disjoint. For the recursive construction of the tree decomposition, 
we need to execute our balanced separator algorithm for all graphs in $\mathcal{G}_\ell$ in parallel.
We circumvent this obstacle by computing the separator of $G'_x$ (i.e. $G_x - B_{p(x)}$) instead of $G_x$. 
Since $\{G'_x \mid x \in A_\ell\}$ are connected and vertex disjoint by definition, one can apply the
technique based on the low-congestion shortcut framework to compute the separators of graphs $G'_x$
for all $x \in A_\ell$ in parallel. Each separator $S'_x$ for $G'_x$ is easily transformed 
into the separator $S_x$ for $G_x$ by adding all vertices in $V(G_x) \cap V(B_p(x))$. The bag $B_x$ is 
defined as $B_x = V(G_x) \cap (\bigcup_{x' \sqsubseteq x} S_{x'}) = B_{p(x)} \cup S'_x$, 
and thus its size is bounded by $O(\TW{G}^2 \log n)$. Our main theorem is stated as follows:

\begin{theorem} \label{thm:treedecomposition}
For a given graph, $G=(V, E)$, there exists a tree decomposition algorithm in the CONGEST model, which constructs a tree decomposition,
$\Phi = (T, \{B_x\}_{x \in V(T)})$, of width $O(\TW{G}^2 \log n)$ whp. 
The depth of $T$ is $O(\log n)$ and the running time of the algorithm 
is $\tilde{O}(\TW{G}^2D + \TW{G}^{3})$ rounds.
\end{theorem}

The complete
details of our distributed implementation, including the formal correctness proof of our tree decomposition algorithm, are given in Appendix~\ref{appendix:DistTreeDecomposition}.

\section{Distributed Distance Labeling in Low-Treewidth Graphs}
\label{sec:distancelabeling}
The proofs of all the lemmas and theorems for this section are deferred to Appendix~\ref{appendix:distancelabeling}.
\subsection{Outline}
Consider the weighted and directed input graph $G = (V(G), E(G))$ with edge cost function $c_G : E(G) \to \mathbb{N}$. The 
distance labeling problem is formally defined as follows:
\begin{definition}[\textbf{Distance Labeling (\textsf{DL})}] Distance labeling consists of 
    a labeling function $\Dlabel_G : V(G) \to \{0, 1\}^{\ast}$ that depends on the input graph $G$ (which can be directed and weighted), 
    and a common \emph{decoder} function $\Ddec : \{0, 1\}^{\ast} \times \{0, 1\}^{\ast} \to \mathbb{N}$. The decoder returns the distance $d_G(u, v)$ from two labels $\Dlabel_{G}(u)$ and $\Dlabel_{G}(v)$. The problem requires that each node $v \in V(G)$ outputs its label $\Dlabel_G(v)$.  
\end{definition}

Our distributed implementation of distance labeling adopts a similar approach to the algorithm of 
Li and Parter~\citep{LP19} for planar graphs, whose structure is a slightly modified version of the 
distance labeling scheme by Gavoille et al.~\citep{GPPR2004}. Our implementation is a recursive 
algorithm utilizing tree decomposition, and can roughly be stated as follows: Let $G$ be a weighted 
directed input graph, $\Phi = (T, \{B_x\}_{x \in V(T)})$ be the (rooted) tree 
decomposition of $G$ constructed by the algorithm of Theorem~\ref{thm:treedecomposition}. 
The algorithm recursively and independently constructs distance labels for each graph in $\Gcal_1$, 
and then each node 
$u$ in $G = G_{\psi}$ learns the distances from/to all of the nodes in $B_{\psi}$ and stores them in the 
label of $u$ constructed in $G_x$ (where $G_x$ is defined in Sec.~\ref{subsec:disttreedecomosition}). Consider computing the distance from $u$ to $v$. If the shortest path form $u$ to $v$ does not contain any vertex in $B_{\psi}$, the distance is obtained by the label of $u$ and $v$ for $G_x$. Otherwise, it suffices to take the minimum of the distance from $u$ to $s$ plus that from $s$ to $v$ for every 
$s \in B_{\psi}$, which can computed from the labels of $u$ and $v$. 

We formally define the distance labeling constructed by our CONGEST algorithm.
For simplicity of presentation, we assume that the edge cost function $c_G$ is a mapping 
from $V(G) \times V(G)$ to $\mathbb{N} \cup \{\infty\}$, where we define $c_G(u, v) = \infty$ if 
$(u, v) \not\in E(G)$. A \emph{distance set} 
$d_G(u, X)$ for $u \in V(G)$ and $X \subseteq V(G)$ is defined as the set of tuples 
$(u, v, d_G(u, v))$  and $(v, u, d_G(v, u))$ for all $v \in X$. We also define $B^{\uparrow}_{\Phi}(u) = 
\bigcup_{x' \sqsubseteq \St_{\Phi}(u)} B_{x'}$.
The label $\Dlabel_{G}(u)$ is defined as $\Dlabel_{G}(u) = d_{G}(u, B^{\uparrow}_{\Phi}(u))$.
The decoder function $\Ddec$ is defined as follows:
\begin{align*}
    \Ddec(\Dlabel_{G}(u), \Dlabel_G(v)) = \min_{s \in B^{\uparrow}_{\Phi}(u) \cap B^{\uparrow}_{\Phi}(v)} d_G(u, s) + d_G(s, v).
\end{align*}
Using the tree decomposition algorithm of Theorem~\ref{thm:treedecomposition}, the label size is bounded by 
$\tilde{O}(\TW{G}^2)$ bits. The lemma below guarantees the correctness of this labeling scheme.
\begin{lemma} \label{lma:correctnessDL}
For any $u, v \in V(G)$, $\Ddec(\Dlabel_{G}(u), \Dlabel_{G}(v)) = d_G(u, v)$ holds.
\end{lemma}

\subsection{Distance Labeling Construction}

We explain the construction of $\Dlabel_G(u)$ for all $u \in V(G)$ in the CONGEST model. First we 
introduce the graph $H_{x}$ associated with each $B_{x}$ as follows:
\begin{itemize}
    \item If $x \in V(T)$ is a leaf node, we define $H_{x} = G_x$. 
    \item \sloppy{Otherwise, $V(H_x) = B_{x}$. An edge $(u, v)$ is contained in $E(H_x)$ if and only if 
    $d_{G_x}(u, v)$ is 
    finite or $(u, v) \in E(G)$ holds. The edge cost $c_{H_x}(u, v)$ is defined as $c_{H_x}(u, v) = 
    \min\{c_G(u, v), \min_{i \in \CHT(x)} d_{G_{x \bullet i}}(u, v)\}$.}
\end{itemize}
The key properties of the graph $H_x$ are stated in the following two lemmas. 
\begin{lemma} \label{lma:distH}
For any $u, v \in V(H_x)$, $d_{H_x}(u, v) = d_{G_x}(u, v)$ holds.
\end{lemma}

\begin{lemma} \label{lma:distUpdate}
Let $u$ and $v$ be any two vertices in $V(G_{x \bullet i}) \cup B_x$ for some $i \in \CHT(x)$. Then the following equality holds.
\begin{align*}
d_{G_x}(u, v) &= 
\min \{d_{G_{x \bullet i}}(u, v), \\ 
& \quad \quad \min_{s, s' \in V(H_x)} (d_{G_{x \bullet i}}(u, s) + d_{H_x}(s, s')  + d_{G_{x \bullet i}}(s', v))\}. 
\end{align*}
\end{lemma}

The construction of the labels follows a bottom-up recursion over the decomposition tree $T$.
More precisely, the proposed algorithm constructs $\Dlabel_{G_x}(u)$ for all $u \in V(G_{x})$,
provided that $\Dlabel_{G_{x \bullet i}}(u)$ for all $u \in G_{x \bullet i}$ and $i \in \CH(T, x)$ are 
available. The outline of the algorithm is as follows:
\begin{enumerate}
   \item If $x$ is a leaf in $T$, each node $u \in V(G_{x})$ broadcasts the information of the edges incident to $u$ in $G_x$ to the nodes in $G_{x}$, i.e., each node $u$ knows the entire information 
   of $G_{x}$. Since the collection of the graphs for every leaf $x$ is not vertex disjoint, we implement this process by introducing a slightly generalized version of the 
   part-wise aggregation (see Appendix~\ref{appendix:subgraphoperation} for the details). By solving
   the all-pairs shortest paths problem locally, $u$ obtains the label $\Dlabel_{G_{x}}(u)$.  
   If $x$ is not a leaf, the algorithm executes steps 2-4.
   \item  For all $i \in \CHT(x)$, the algorithm 
   recursively constructs the distance labeling $\Dlabel_{G_{x \bullet i}}(u)$ for $u \in V(G_{x \bullet i})$ utilizing
   $\Phi' = (T(x \bullet i), \{B_{x'}\}_{x \bullet i \sqsubseteq x'})$ as the tree decomposition of $G_{x \bullet i}$.
   Since the node $x \bullet i$ is the root of 
   $T(x \bullet i)$, $B_{x \bullet i} \subseteq B^{\uparrow}_{\Phi'}(u)$ trivially holds. That is, 
   $\Dlabel_{G_{x \bullet i}}(u) = d_{G_{x \bullet i}}(u, B^{\uparrow}_{\Phi'}(u))$ necessarily contains the information 
   of the distance
   set $d_{G_{x \bullet i}}(u, B_{x \bullet i})$. Thus, each node $u \in B_{x}$ can identify 
   the edges incident to $u$ in $H_{x}$. This process is done via local computation.
   \item Each node $u \in B_{x}$ broadcasts the set of edges incident to $u$ in $H_{x}$ to all nodes in $G_{x}$, which is also implemented by the generalized part-wise aggregation.
   \item Using the information received at step 3, each node $u \in V(G_{x})$ locally knows $H_{x}$.
   Following the formula of Lemma~\ref{lma:distUpdate}, $u$ updates the distance set 
   $d_{G_{x \bullet i}}(u, B^{\uparrow}_{\Phi'}(u))$ to $d_{G_{x}}(u, B^{\uparrow}_{\Phi'}(u))$, and learns the distance set $d_{G_x}(u, B_x)$.
\end{enumerate}
We state the following theorem.
\begin{theorem} \label{thm:dlabel}
Let $G = (V(G), E(G))$ be any directed graph with edge cost function $c:E(G) \to \mathbb{N}$.
Then there exists a randomized CONGEST algorithm that solves $\textsf{DL}$ in 
$\tilde{O}(\TW{G}^2D + \TW{G}^5)$ rounds with probability at least $1 - 1/n^{9}$. The label size of each node 
is $O(\TW{G}^2\log^2n)$ bits. 
\end{theorem}

\section{Stateful walks}
\label{sec:statefulwalk}

\subsection{Definition}

Let $G= (V(G), E(G), \gamma_G)$ be a directed multigraph, where $\gamma_G$ is a mapping from each element in $E(G)$ to an ordered pair in $V(G) \times V(G)$ (as $G$ is a multigraph, we cannot assume that elements in $E(G)$ are of the form $(u,v)$, and thus $\gamma_G$ is required).
A walk $w = e_1, e_2, \dots, e_{\ell}$ in $G$ is a sequence of edges in $E(G)$ 
such that for any $i \in [1, \ell -1]$, $\gamma_G(e_{i})[1] = \gamma_G(e_{i+1})[0]$ holds. 
To describe the vertices in the walk explicitly, it can also be represented as an alternating 
sequence of vertices and edges, $w = u_0, e_1, u_1, e_2, u_2, \dots, e_{\ell}, u_{\ell}$, such that $\gamma_G(e_1)[0] = u_0$ and $\gamma_G(e_i)[1] = u_i$ holds for any $i \in [1, \ell]$. We denote by $W_G$ the set of all finite-length walks in $G$, and 
also denote the set of all walks from vertex $s$ to vertex $t$ in $W_G$ by $W_{G}(s, t)$.
The walk of length zero in $W_G$ is denoted by $\phi$. 
For two walks $x, y \in W_G$ such that the last vertex of $x$ and the first vertex of $y$ are the same, their concatenation is denoted by $x \circ y$. If $y$ (resp. $x$) is a walk of 
length one consisting of an edge $e \in E(G)$, we use the notation $x \circ e$ 
(resp. $e \circ y$). A walk $w$ is often treated as a subgraph of $G$, i.e., $V(w)$ and $E(w)$ respectively denote the vertices and edges in $w$.

A \emph{walk constraint} is 
a subset $C \subseteq W_G$. That is, a walk constraint limits the set of graph walks 
to the subset $C$. Given a walk constraint $C$, we denote by $W_{G, C}(s, t)$ the set of all walks 
from $s$ to $t$ in $C$. If $W_{G, C}(s, t)$ is not empty, we say that $t$ is \emph{$C$-reachable} 
from $s$ in $G$. The \emph{$C$-distance} from $s$ to $t$, denoted by $d_{G, C}(s, t)$, is the shortest (weighted) length of all the walks in $W_{G, C}(s, t)$.  We consider a variation of the 
single-source shortest paths problem, which requires that for a given constraint $C$ and source node $s$
each node $v \in V(G)$ outputs the shortest walk from $s$ to $v$ in $C$, as well as its length 
$d_{G, C}(s, v)$. 
As discussed previously, this problem becomes meaningful only when $C$ is provided to the nodes of 
the graph in an \emph{implicit} and \emph{distributed} manner. To formally define the above, we present 
the notion of a \emph{stateful walk constraint}, followed by an intuitive description of the definition.
\begin{definition}[Stateful Walk Constraint] \label{def:statefulwalk}
Let $Q$ be any set containing two special elements $\perp$ and $\Auxinit$. A 
walk constraint $C \subset W_G$ is called \emph{stateful} if it contains $\phi$ and 
there exist a function $\Mem_C : W_G \to Q$ and a tuple of transition functions 
$\boldsymbol{\Trans_C} = (\Trans_{C, e})_{e \in E(G)}$ associated with each edge $e \in E(G)$, where
$\delta_{C, e}$ is a mapping from $Q$ to $Q$, satisfying the following three conditions:
\begin{enumerate}
    \item $\Mem_C(w) \neq \perp$ if and only if $w \in C$. In addition, $\Mem_C(w) = \Auxinit$ if and only if $w = \phi$.
    \item For any $w \in W_G$ terminating at $u$ and $e = (u, v) \in E(G)$, $\Trans_{C, e}(\Mem_C(w)) = \Mem_C(w \circ e)$ holds.
    \item For any $e \in E(G)$, $\Trans_{C,e}(\perp) = \perp$.
\end{enumerate}
\end{definition}
We omit the subscript $C$ of $\Mem$ and $\Trans$ when it is clear from the context.
The rough intuition of the definition above is as follows: Every walk $w \in W_G$ has a state in $Q$ (which is referred to as the \emph{state} of $w$ hereafter). Starting from the length-zero walk $\phi$, which has the special state $\Auxinit$, the state of the walk 
changes as the length of the walk increases. The function $\delta_{e}$ determines how the state of a given 
walk $w$ changes by appending edge $e$ to its tail. The second condition implies that 
the state of $w \circ e$ is determined only by the state of $w$ and the edge $e$ and is independent of any other feature of $w$. The state $\perp$ is a ``reject'' state, which implies $w$ 
does not satisfy the constraint $C$ (i.e., $w\notin C$). Condition 3 implies that once a walk $w$ does not satisfy $C$ (i.e., $M_C(w)=\perp$), no addition of edges to $w$ can make it satisfy $C$. 
Let $\Eout{G}(u)$ be the set of outgoing edges from $u$ in $G$.
Assuming each node $u$ knows the transition function $\delta_{C, e}$ for all $e \in \Eout{G}(u)$, a stateful walk constraint implies that each node $u$ can \emph{locally} decide the state of a walk $w$ leaving $u$ only from the state of the incoming prefix of the walk.
For a stateful walk constraint $C$, a walk with a state other than $\perp$ (i.e. a walk in $C$) 
is called a \emph{stateful} walk. We now present two concrete examples of stateful walks.
\begin{example}[$c$-Colored Walk] \label{example:color}
Here we assume edges have colors, and we are interested in walks where no two consecutive 
edges are monochromatic. Consider the edge label domain $\Sigma$ of cardinality $c$ (i.e. 
color palette), and an assignment $f : E(G) \to \Sigma$ of a color to each edge. 
A \emph{$c$-colored walk} $w = e_1, e_2, \dots, e_{\ell}$ is a walk satisfying $f(e_{i}) \neq f(e_{i+1})$ 
for all $1 \leq i \leq \ell - 1$. The set of all $c$-colored walks $\Ccolor{c} \subseteq W_G$
is a stateful walk constraint associated with the following triple $(Q, \Mem, \boldsymbol{\Trans})$: 
The state domain is $Q = \Sigma \cup \{\perp, \Auxinit\}$. For any $w \in W_G$, $\Mem(w)$ represents 
$f(e)$ for the last edge $e$ in $w$ if $w \in \Ccolor{c}$, $\Auxinit$ if
$w = \phi$, and $\perp$ otherwise. The state $\Trans_e(q)$ is $f(e)$ if $q \neq f(e)$, 
and $\perp$ otherwise.
\end{example}
\begin{example}[count-$c$ Walk] \label{example:count}
Here we assume edges are assigned a binary value (i.e., zero or one), and we are interested in walks 
that contain at most $c$ edges of value one. A \emph{count-$c$ walk} 
$w = e_1, e_2, \dots, e_{\ell}$ is a walk satisfying $\sum_{i \in [1, \ell]} f(e_i) \leq c$.
The set of all count-$c$ walks, $\Ccount{c}$, is a stateful walk constraint associated with the following triple $(Q, \Mem, \boldsymbol{\Trans})$: The state domain is defined as $Q = [0, c] \cup \{\perp, \Auxinit\}$. 
For any 
$w \in W_G$, $\Mem(w)$ represents $\sum_{e \in w} f(e)$ if it is within the range $[0, c]$, $\Auxinit$ if $w = \phi$, and $\perp$ otherwise. The state $\Trans_e(q)$ is $q + f(e)$ if $q \not\in \{\Auxinit, \perp\}$ and $q + f(e) \leq c$ hold, $f(e)$ if $q = \Auxinit$, and $\perp$ otherwise. 
\end{example}
As seen in the examples above, the specification of the function $\delta_e$ typically 
relies only on the edge label $f(e)$, but this characteristic is not mandatory. 
\paragraph{Subsets of stateful walk constraints}
Let $C$ be a stateful walk constraint associated with the triple 
$(Q, \Mem, \boldsymbol{\delta})$. We denote by $C(q)$ the set of all walks with state $q$ in $C$, 
and define $C(Q') = \cup_{q \in Q'} C(q)$ for $Q' \subseteq Q$. 
By definition, $C(Q')$ for any $Q' \subseteq Q$ is a walk constraint (but not 
necessarily stateful). For example, while the stateful constraint for count-$c$ walks considers all walks of count \emph{at most $c$}, we can define \emph{exact} count-$c$ walks (where the count is exactly $c$) as $C(c) \subset C$. 

\subsection{Finding Stateful Walks}
\label{subsec:auxiliarygraph}
In this section, we show how to reduce the problem of finding shortest stateful walks to the problem of finding 
\emph{unconstrained} shortest walks.
We present a general framework for reducing the constrained version of the shortest paths problem,
for any stateful constraint $C \subseteq W_G$, into the unconstrained version in some
auxiliary directed graph $G_{C}$.
The construction of $G_{C}$ is defined as follows:
\begin{itemize}
    \item \sloppy{
    Let $U_Q(u) = \{(u, i) \mid i \in Q\}$, and define $V(G_C) = \bigcup_{u \in V(G)} U_Q(u)$ (i.e., $V(G_C) = V(G) \times Q$).}
    \item $((u, i), (v, j)) \in E(G_{C})$ if and only if one of the following conditions holds: (1) There exists an edge $e = (u, v) \in E(G)$ satisfying $\Trans_{C,e}(i) = j$. (2) $u = v$, $i \neq \perp$, and $j = \perp$ hold.
    \item If the input graph $G$ is weighted (by an edge-cost function $c : E(G) \to \mathbb{N}$),
    for any $u, v \in V(G)$, assign the cost $c(u, v)$ to any edge 
    $E(G_C) \cap (U_Q(u) \times U_Q(v))$.
\end{itemize}

The intuition of the above construction is that we wish to break down the 
vertex $u$ into $U_Q(u)$ in order to distinguish walks entering $u$ with different states. The 
vertex $(u, i)$ can be seen as the arrival vertex of any walk $w$ to $u$ with state $i$. Since we add an edge between $(u, i)$ and $(v, j)$ if and only if $\Trans_{C, e}(i) = j$ holds, the walk $w \circ e$, which has state $\Trans_{C, e}(i) = j$ in the original graph $G$, 
always reaches $(v, j)$. An illustration of this construction is presented in Figure~\ref{fig:graph-construction}. Note that the second condition is introduced in order to bound 
the diameter of $\Comm{G_{C}}$ by $O(D)$. The distance from any node in $U_Q(u)$ to any node in $U_Q(v)$ 
is at most $d_{\Comm{G}}(u, v) + 2$ because there exists a walk from $(u, \perp)$ to $(v, \perp)$ of 
length $d_{\Comm{G}}(u, v)$ (recall that condition 2 of Definition~\ref{def:statefulwalk} implies $((u, \perp), (v, \perp)) \in E(G_C)$ for any $e = (u, v) \in E(G)$).  We state the following lemma.
\begin{lemma} \label{lma:shorteststatefulwalk}
Let $G = (V(G), E(G), \gamma_G)$ be any multigraph with edge-cost function 
$c:E(G) \to \mathbb{N}$, and $C \subseteq W_G$ be a stateful walk constraint with associated triple 
$(Q, \Mem, \boldsymbol{\Trans})$. There exists a walk $w$ of weighted length $x$ from $s$ to $t$ with state $q$ 
($s, t \in V(G)$, $q \in Q \setminus \{\perp\}$) in $C$ 
if and only if there exists a walk $w'$ of weight $x$ from $(s, \Auxinit)$ to $(t, q)$ in $G_C$. 
\end{lemma}
The proof of the lemma is deferred to Appendix~\ref{appendix:statefulwalk}. For any stateful walk constraint $C$, its state $q$, and two vertices $s, t \in V_G$, this lemma allows us to compute the shortest walk in $C(q)$ from $s$ to $t$ (and its distance) by computing the directed shortest path from $(s, \Auxinit)$ to $(t, q)$ in $G_C$. Letting 
$p_{\max}$ be the maximum edge multiplicity of the original graph $G$, it is easy to simulate 
the execution of any CONGEST algorithm for $\Comm{G_{C}}$ on top of the original communication graph $\Comm{G}$ with $O(|Q| \cdot p_{\max})$-round overhead: Each node $v \in V(G)$ is responsible for the simulation of the nodes in $U_Q(v)$. Consider the subgraph $H(u, v)$ of $G_{C}$ induced by $U_Q(u) \cup U_Q(v)$. Each node in this subgraph has at most $p_{\max}$ outgoing edges. Thus, the total number of edges 
in $H(u, v)$ is at most $2p_{\max}|Q|$. A single communication round over the links in $E(\Comm{H(u, v)})$ can be achieved by $O(|Q|p_{\max})$ communication rounds over the edge $(u, v) \in E(\Comm{G})$. 
The total number of nodes in $\Comm{G_C}$ is $|Q|n$, and the diameter of $G_{C}$ is $O(D)$. 
It is easy to show that the treewidth of $G_{C}$ is bounded by $O(|Q| \cdot \TW{G})$: Given a tree decomposition of $G$, we replace each vertex $v$ in each bag by $U_Q(v)$. The resulting decomposition 
is obviously a tree decomposition of $G_{C}$ and the bag size is multiplied by $|Q| + 1$.
Consequently, any $f(n, D, \TW{G_C})$-round algorithm in $G_C$ is simulated on the top of $\Comm{G}$
within $O(|Q|p_{\max}f(n, D, \TW{G}(|Q| + 1))$ rounds.
 
In our applications, we are interested in the constrained version of distance labeling schemes, 
which is formalized as follows:
\begin{itemize}
    \item \textbf{Constrained distance labeling (\textsf{CDL($C$)})}: Let $C$ be a stateful walk constraint with associated triple $(Q, \Mem, \boldsymbol{\Trans})$. It consists of a 
    labeling function $\Slabel_{G,C}: V(G) \to \{0, 1\}^{\ast}$, that depends on the input graph $G$, and a common \emph{decoder} function $\Sdec_{C} : Q \times \{0, 1\}^{\ast} \times \{0, 1\}^{\ast} \to \mathbb{N}$. Both functions must satisfy $\Sdec_C(q, \Slabel_{G,C}(u), \Slabel_{G,C}(v)) = d_{G, C(q)}(u, v)$ for any $u, v \in V(G)$ and $q \in Q$. The problem requires that each node $v \in V(G)$ outputs its label $\Slabel_{G, {C}}(v)$. 
\end{itemize}
Note that the input graph $G$ can be directed and weighted. 
The problem \textsf{CDL}($C$) in $G$ is solved by any algorithm for (standard) distance labeling 
in $G_C$. Since a node $u \in V(G)$ simulates all nodes $U_Q(u) \subseteq V(G_C)$, after
the construction of the standard distance labeling for $G_C$, $u$ has the set of
labels $\{\Dlabel_{G_C}((u, i)) \mid i \in Q\}$, which we take as the output of \textsf{CDL}($C$). 
This is because one can obtain the $C(q)$-distance from $u$ to $v$ by computing 
$\Ddec(\Dlabel_{G_C}((u,\Auxinit)), \Dlabel_{G_C}((v, q)))$.
To solve \textsf{CDL($C$)} in a low-treewidth graph, $G$, we use the algorithm of Theorem~\ref{thm:dlabel}.
Running it on $G_C$, we obtain an algorithm for \textsf{CDL($C$)}. Consequently, 
we obtain the following theorem.
\begin{theorem} \label{thm:slabel}
Let $G = (V(G), E(G), \gamma_G)$ be a multigraph of maximum edge multiplicity $p_{\max}$, 
$c:E(G) \to \mathbb{N}$ be an edge cost function of $G$, and 
$C$ be a stateful walk constraint with associated triple $(Q, \Mem, \boldsymbol{\Trans})$. 
Then there exists a randomized CONGEST algorithm that solves $\textsf{CDL}(C)$
in $\tilde{O}(|Q|p_{\max}((|Q|\TW{G})^{2}D + (|Q|\TW{G})^4)))$ rounds whp.
\end{theorem}
While distance labeling only outputs the distance, it can be easily transformed into an algorithm
for finding the shortest stateful walk between any pair of vertices $s, t \in V(G)$.
\begin{corollary} \label{corol:pathconstruction}
Let $G = (V(G), E(G), \gamma_G)$ be a multigraph of maximum edge multiplicity $p_{\max}$, 
$c:E(G) \to \mathbb{N}$ be an edge cost function of $G$, and 
$C$ be a stateful walk constraint with associated triple $(Q, \Mem, \boldsymbol{\Trans})$. 
Then there exists a randomized CONGEST algorithm for constructing the shortest walk 
$w$ in $W_{G, C(q)}(s, t)$ for any given $v$ and $q \in Q$ whp. Each node in $V(w)$ outputs the distance 
along $w$ from $s$ and its predecessor in $w$. The running time of the algorithm is $\tilde{O}(|Q|p_{\max}((|Q|\TW{G})^{2}D + (|Q|\TW{G})^4)))$ rounds.
\end{corollary}

\begin{figure}[ht]
	\begin{center}
		\includegraphics[clip, scale=0.25]{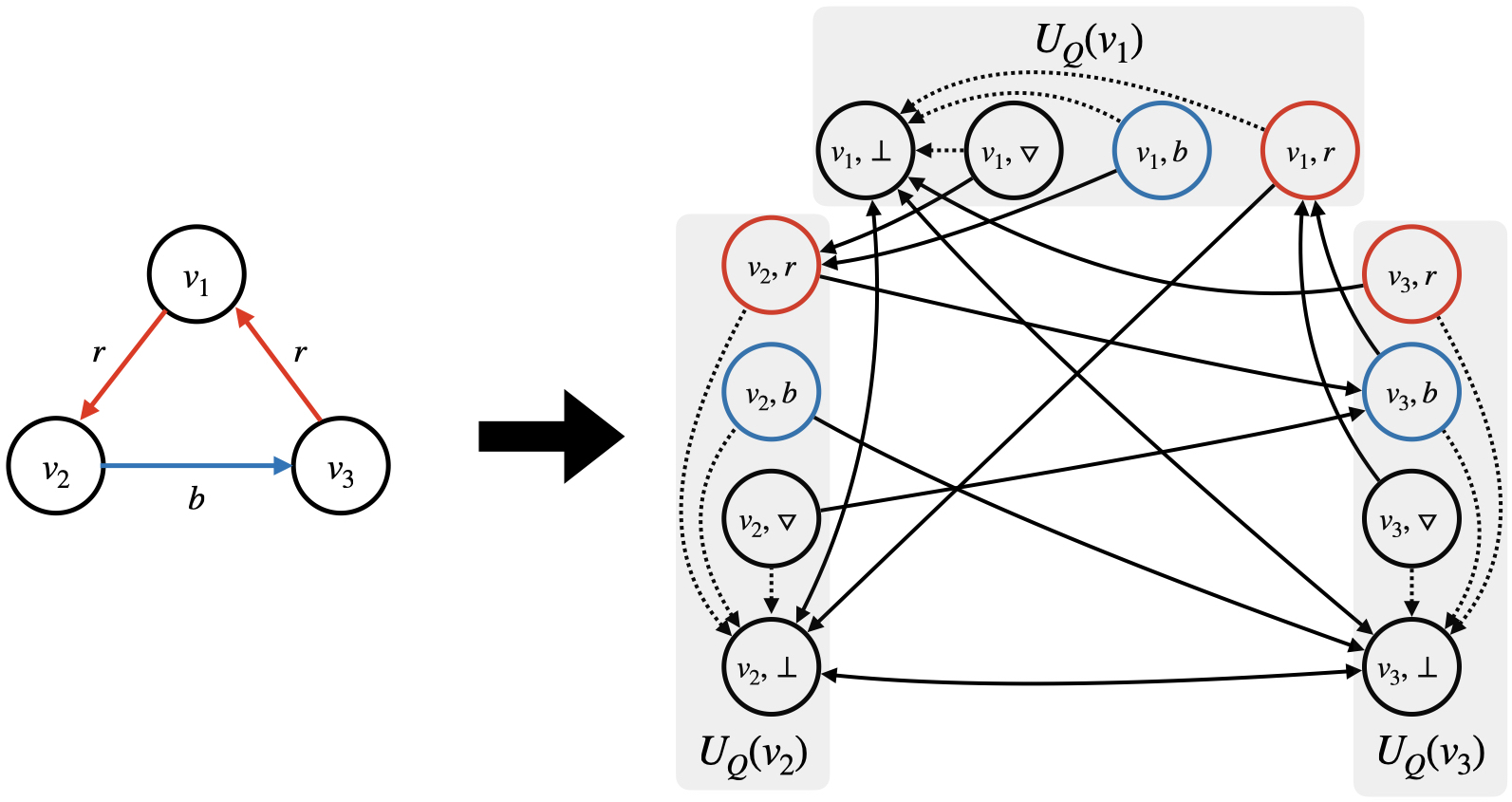}
	\end{center}
	\caption{\label{fig:graph-construction} An illustration of the construction of the graph $G_{C}$ (right) from a graph $G$ (left). Here we consider the $c$-Colored walk problem, where $Q=\{r,b, \perp, \Auxinit\}$. Note that $G_C$ has no edge labels.
	The dashed lines are edges added according to condition (1), while the solid lines are edges added according to condition (2). For simplicity, the above example is not weighted. However, if, for example, it holds that $c((v_2,v_3))=x$ in $G$ then the weight of all edges going from $U_Q(v_2)$ to $U_Q(v_3)$ (left to right) in $G_c$ will have weight $x$.}
\end{figure}

\section{Exact Bipartite Maximum Matching}
\label{sec:matching}
Let $G$ be an undirected unweighted graph.
A matching $M \subseteq E(G)$ is a set of edges such that any two distinct edges do not
share an endpoint. Given a matching $M$ of a graph $G$, we say that a vertex $u$ is \emph{unmatched} if it is 
not an endpoint of any matching edge. The maximum matching problem requires the 
algorithm to output the maximum cardinality matching (by marking the edges in the solution). 
The maximum matching problem is known to be reducible to the task of finding an 
\emph{augmenting path}, a simple path connecting two unmatched nodes where matching edges 
and non-matching edges appear alternately. Once an augmenting path is found, by flipping matching edges and non-matching
edges, the size of the matching increases by one. The maximum matching is obtained by iterating this augmentation 
process until the current matching does not have any augmenting path. 

An augmenting path can be seen as a simple $2$-colored walk whose endpoints are both unmatched vertices, and thus it fits naturally into our stateful-walk framework (more precisely, the construction of $\Ccolor{2}$-distance labeling following Example~\ref{example:color} and Theorem~\ref{thm:slabel}).
This idea is incorrect for general graphs because the shortest $2$-colored walk is not necessarily simple, but is valid for bipartite graphs: It is well-known that any shortest $2$-colored walk is simple in bipartite graphs. 
However, we still have a few hurdles to overcome. 
The first issue is how does each unmatched node detect if there exists an augmenting path starting from itself? Since there might exist $\Omega(n)$ unmatched nodes, the trivial solution where all unmatched nodes broadcast their own $\Ccolor{2}$-distance labels is very costly. The second issue is how to speed up the iterations of the matching update. The trivial sequential update takes $\Omega(n)$ iterations. We resolve these issues by a divide-and-conquer approach utilizing balanced separators. 
The key observation is 
that if the maximum matching is already computed for each connected component in $G - S$ independently (where $S \subseteq V(G)$ is any vertex subset), it suffices to check only the 
existence of augmenting paths with at least one endpoint in $S$. 
This observation is proved in the following simple proposition, which is a special case 
of a more general theorem presented in \citep{IOO18}:
\begin{proposition}[Iwata et al. \citep{IOO18}]
\label{prop:maximummatching}
Let $G$ be any (undirected and unweighted) graph $G$, $U \subseteq V(G)$ be a vertex subset, 
and $\mathcal{H} = \{H_1, H_2, \dots, H_N\}$ be the set of connected components of $G - U$. 
Assume that the maximum matching $M_i$ for each connected component $H_i \in \mathcal{H}$ is already computed. Then, for any $v \in V \setminus U$, the size of the maximum matching of the graph $G - (U \setminus \{v\})$ is at most $|\bigcup_{1 \leq i \leq N} M_i| + 1$. Any augmenting path in $G - (U \setminus \{v\})$ starts from $v$.
\end{proposition}

This proposition naturally induces a divide-and-conquer approach for bipartite maximum matching in 
low-treewidth graphs. That is, computing a balanced separator $S$, we recursively compute the maximum 
matching of each connected component of $G - S$. To obtain the maximum matching for the whole graph $G$, it suffices
to consider only the augmenting paths starting from $S$. We state the following theorem (the proof is deferred to Appendix~\ref{appendix:matching}).

\begin{theorem} \label{thm:matching}
There exists a randomized CONGEST algorithm that computes the maximum matching for any bipartite graph $G$ 
in $\tilde{O}(\TW{G}^{4}D + \TW{G}^7)$ rounds whp.
\end{theorem}

\section{Computing Girth}
\label{sec:girth}
Let $G$ be a simple and (positively) weighted graph, which could be directed or undirected. 
The \emph{girth} of $G$ (denoted by $g$) is the weight of the shortest\footnote{By shortest, we mean smallest weight.} simple cycle in $G$. For directed graphs, it is relatively easy to reduce the computation of girth
to the distance labeling construction: The length of the shortest cycle containing a directed edge 
$(u, v) \in E(G)$ is determined by computing the distance from $v$ to $u$ in $G$, which is obtained by 
exchanging the labels $\Dlabel_G(u)$ and $\Dlabel_G(v)$. To compute the girth, it suffices to 
execute this task for all edges, and take the minimum over all of the computed cycle lengths. In our setting,
the running time of this algorithm is $\tilde{O}(\TW{G}^2D + \TW{G}^5)$ rounds.

The case for undirected graphs is more challenging because the shortest path from 
$v$ to $u$ can contain the edge $(u, v)$, while such a case does not occur in directed graphs.
This section provides a CONGEST algorithm that computes $g$ in an undirected graph $G$
using our framework. Let $\Sigma = \{0, 1\}$ be the edge label domain. Recall that a walk $w$ is called \emph{exact count-1} if $w$ contains exactly one edge with label one. Since $\Ccount{1}(1)\subset \Ccount{1}$ is obviously the set of all exact count-1 walks, Theorem~\ref{thm:slabel} allows us to compute the shortest length of exact count-1 walks from $u$ to $v$ for any two nodes $u, v \in V(G)$ using their labels (i.e., $\Sdec_{\Ccount{1}}(1, \Slabel_{\Ccount{1}}(u), \Slabel_{\Ccount{1}}(v))$ is the shortest length of exact count-1 walks from $u$ to $v$). The following lemma is the key for our girth algorithm. 
\begin{lemma} \label{lma:count-one-girth}
Any shortest exact count-1 walk $w$ starting and terminating at the same vertex $v$ contains a simple cycle, and thus the weight of $w$ is at least $g$. 
\end{lemma}

Note that the above holds for \emph{any} assignment of binary edge labels. Assume a labeling
function $f$ such that some shortest cycle $R = e_0, e_1, \dots, e_{g-1}$ has exactly one edge $e_i \in E(R)$ which satisfies $f(e_i) = 1$. Each node $u$ computes the length of the shortest exact count-1 walk from $u$ to $u$. Let us denote this by $g(u)$. 
As explained above, this is possible by using \textsf{CDL}($\Ccount{1}$)
(note that the shortest length of exact count-1 walks (i.e., $\Ccount{1}(1)$-distance) from $u$ to $u$ is computed locally by the label $\Slabel_{\Ccount{1}}(u)$). As $R$ is a shortest cycle in $G$ with exactly one edge labeled "1", by Lemma~\ref{lma:count-one-girth}, $g(v) = g$ holds for every $v \in V(R)$. Thus, we can compute the girth $g = \min_{u \in V(G)} g(u)$ by standard aggregation over all nodes. 

The final challenge is how to obtain the edge label function $f$ satisfying the above condition. However, this can be resolved by a probabilistic label assignment. Let $F$ be the set of the edges, $e \in E(G)$, such that $e$ is covered by at least one shortest cycle. Note that if exactly one edge in $F$ has label one, the condition holds for at least one shortest cycle. To guarantee that the condition holds with constant probability, it suffices to assign each edge with label one (independently of other edges) with probability $p = \Theta(1/|F|)$. Repeating this process (assigning random labels, computing $g$, keeping the minimum value of $g$) a logarithmic number of times, we can amplify the success probability to $1-1/n$. While the value of $p$ is unknown to the algorithm, it can be estimated by a standard doubling technique.
Consequently, we state the following theorem (the complete argument and proofs are presented in Appendix~\ref{appendix:girth}).

\begin{theorem} \label{thm:girth}
There exists a randomized CONGEST algorithm that computes the girth, $g$, of directed and weighted
graph $G$ in $\tilde{O}(\TW{G}^{2}D + \TW{G}^5)$ rounds whp.
\end{theorem}



\newcommand{\etalchar}[1]{$^{#1}$}

\newpage

\appendix

\section{Preliminaries for the Appendix}
\label{appendix:preliminaries}

\subsection{Part-wise Aggregation for Near Disjoint Collection}
\label{appendix:subgraphoperation}

We present a generalized version of part-wise aggregation, which is applicable to a certain kind of edge-disjoint but not necessarily vertex-disjoint subgraphs. 
A collection of connected subgraphs $\mathcal{H} = \{H_1, H_2, \dots, H_N\}$ ($H_i \subseteq G$) is called
a \emph{near disjoint collection} of $G$ if it satisfies the following two conditions: 
\begin{itemize}
    \item For any edge $(u, v) \in E(G)$, either $u$ or $v$ is contained in at most one subgraph in $\mathcal{H}$.
    \item Let $V'(H_i) \subseteq V(H_i)$ be the set of nodes which do not belong to any other 
    subgraph in $\mathcal{H} \setminus \{H_i\}$. For any $i \in [1, N]$,
    the subgraph $H'_i$ induced by $V'(H_i)$ is connected.
\end{itemize}
We extend the applicability of subgraph aggregation to any near disjoint collection of $G$.
\begin{lemma}
	\label{lma:SA}
	For any undirected graph $G$ and its near disjoint collection $\mathcal{H}$, subgraph aggregation for $\mathcal{H}$ and any associative binary operator $\oplus$
	can be solved deterministically in $\tilde{O}(\TW{G}D)$ rounds.
\end{lemma}
\begin{proof}
As a preprocessing step, for all $i \in [1, N]$, each node $u \in V(H_i) \setminus V(H'_i)$ sends its input value $x_{u, i}$ to 
one neighbor $v$ that belongs to $V(H'_i)$.
Such a neighbor necessarily exists by the first condition of near disjointness. The node $v$ receiving 
$x_{u, i}$ sets $x_{v, i} \oplus x_{u, i}$ as its new 
input value. Then the algorithm executes PA for $\mathcal{H'} = \{H'_1, H'_2, \dots, 
H'_N\}$. Since $\mathcal{H}'$ is a collection of vertex disjoint connected subgraph, this PA task can be implemented by known algorithm~\citep{HIZ16-2,HHEW18}, whose running time is $\tilde{O}(\TW{G}D)$ rounds.
The output for $H'_i$ is obviously equal to $\bigoplus_{v \in V(H_i)} x_{v, i}$. This output value for 
$H'_i$ is sent back to the nodes in $V(H_i) \setminus V(H'_i)$ in one round.
The total running time is dominated by that for PA, which takes $\tilde{O}(\TW{G}D)$ rounds. 
\end{proof}

Most of the algorithms presented hereafter are constructed by utilizing the (generalized) 
part-wise aggregation (PA) and standard one-round neighborhood communication (i.e., a single round 
of the CONGEST model) as fundamental building blocks. We formalize one-round neighborhood communication as a task of subgraph operations, which we refer to as the short-hand ``SNC.'' That is, SNC for $\mathcal{H} = 
\{H_1, H_2, \dots, H_N\}$ is the task where each node $u \in V(H_i)$ exchanges $O(\log n)$ bit messages with its neighbors in $H_i$. Trivially, SNC can be implemented for any near disjoint collection $\mathcal{H}$ in a single round.

It is known that several fundamental tasks in the distributed setting can be reduced to PA and SNC.
\begin{lemma}[Ghaffari et al.~\citep{GH16-2}, Haeupler et al.~\citep{HL18}, Li~\citep{Li18}, generalized to near-disjoint collections]
    \label{lma:SAapplication}
    \sloppy{
    Let $\mathcal{H} = \{H_1, H_2, \dots, H_N\}$ be a near-disjoint collection of connected subgraphs of $G$. 
    The tasks below are deterministically solvable simultaneously and independently for all $H_i$ 
    by $\tilde{O}(1)$ invocations of PA and SNC running on $G$.}
    \begin{itemize}
        \item Subgraph Rooted Spanning Tree (RST): Given root node $r_i \in V(H_i)$, compute a spanning tree $T_i$ rooted at $r_i$ for each $H_i$. 
        The output of every $v \in V(H_i)$ is the set of all neighbors in $T_i$ as well as a pointer to the parent vertex in $T_i$. The root node outputs a pointer to itself. 
        \item Subtree Aggregation (STA): 
        Assume that each $H_i$ is a rooted tree. Let $\oplus$ 
        be an associative binary operator over a finite set $\mathcal{M}$ of cardinality $\mathrm{poly}(n)$. 
        Given inputs $x_{v,i} \in \mathcal{M}$ for all $v \in V(H_i)$ and $i \in [1, N]$, each vertex $v \in V(H_i)$ outputs the value $\bigoplus_{u \in V(H_i(v))} x_{u, i}$ (recall that $H_i(v)$ is the subtree of $H_i$ rooted at $v$).
        \item Subgraph Leader Election (SLE): Elect a unique leader independently in each subgraph $H_i$. Each
        node $v \in V(H_i)$ has a binary input $x_{v, i} \in \{0, 1\}$ indicating if $v$ is a candidate or not.
        The output of each node in $H_i$ is the ID of the elected leader in $H_i$. The input of the elected leader must be one.
        \item Connected Component Detection (CCD): For each $H_i$, detect all of the connected components of a given subgraph $H'_i \subseteq H_i$. The input of each node $u \in V(H_i)$ is the set of edges
        incident to $u$ in $H'_i$. The algorithm assigns each connected component in $H_i$ with a unique ID,
        which is outputted by the nodes in the component.
        \item Subgraph Broadcast (BCT) : Each subgraph $H_i \in \mathcal{H}$ contains a single source node 
        $u_i$, which broadcasts a message, $m_i$, of $O(\log n)$ bits in $H_i$. The output of the algorithm
        at node $v \in V(H_i)$ is $m_i$. 
\end{itemize}
    In addition, the following task is deterministically solvable by $\tilde{O}(t)$ invocations of 
    PA and SNC in $G$. 
    \begin{itemize}
        \item Subgraph minimum vertex cut (MVC($t$)): Given two disjoint subsets $X_i, Y_i \subseteq 
        V(H_i)$ for 
        each $H_i$, output a \emph{$X_i$-$Y_i$ vertex cut} of $H_i$ of size at most $t$. Each node outputs 1 if it is in the computed cut, and 0 otherwise. If there is no cut of size at most $t$, all nodes in $H_i$ output $-1$.
    \end{itemize}
\end{lemma}
We also use the above three-character shorthands for indicating which algorithm is applied to handle
a given task. For example, consider the statement ``the algorithm performs some 
task (STA)'', here we mean that the task is implemented by subtree aggregation. 

\subsection{Scheduling Multiple Instances of Subgraph Operations}
\label{subsec:multiplescheduling}

\sloppy{
We would like to execute $N$ CONGEST algorithms $\mathcal{A} = \{A_1, A_2, A_3, \dots, A_N\}$ simultaneously and independently in a common network $G$. The \emph{dilation} of $A_i$ in $G$ is the running time of 
$A_i$ when only it is executed in $G$. 
Letting $c_i(e)$ be the number of rounds that $A_i$ uses edge $e$ (i.e. sends a message through $e$),
the congestion of $A_i$ in $G$ is the value $\max_{e \in E(G)} c_i(e)$. Note that 
each algorithm $A_i$ can be a randomized algorithm. Then the dilation and congestion of 
$A_i$ become random variables, and a failure of $A_i$ means that either 
$A_i$ outputs a wrong answer or violates the specified dilation or congestion bound.
The dilation of $\mathcal{A}$ is defined as the maximum dilation over all algorithms in $\mathcal{A}$,
and the congestion of $\mathcal{A}$ in $G$ is defined as $\max_{e \in E(G)} \sum_{1 \leq i \leq N} c_i(e)$. The following general scheduling theorem is known:
}

\begin{theorem}[Ghaffari~\citep{Ghaffari15}] \label{thm:scheduling}
Let $\mathcal{A}$ be a set of (possibly randomized) CONGEST algorithms, such that the dilation and congestion of 
$\mathcal{A}$ in $G$ are respectively bounded by $\delta$ and $\gamma$ whp.  Then there 
exists a CONGEST algorithm for running all algorithms in $\mathcal{A}$ in $\tilde{O}(\delta + \gamma)$
rounds whp.
\end{theorem}

Now we bound the congestion and dilation of PA for a near disjoint collection of $G$. The following
result is known:
\begin{lemma}[Haeupler et al.~\citep{HIZ16-2}, \citep{HHEW18}]
	\label{lma:SAcongestion}
	For any undirected graph $G$ and its near disjoint collection $\mathcal{H}$ of connected subgraphs, 
	PA for $\mathcal{H}$ can be solved deterministically with dilation $\tilde{O}(\TW{G}D)$ and congestion $\tilde{O}(\TW{G})$.
\end{lemma}

This lemma yields better running time bounds for executing a collection of tasks presented in Lemma~\ref{lma:SAapplication}
independently in parallel. Specifically, our algorithms 
execute multiple runs of BCT and MVC. We state the following two corollaries.
\begin{corollary} \label{corol:multivertexcut}
Let $\mathcal{H} = \{H_1, H_2, \dots, H_N\}$ be a near disjoint collection of $G$. 
For each subgraph $H_i$, $h$ pairs of disjoint vertex subsets $X_{i,j}, Y_{i,j} \subseteq 
        V(H_i)$ ($1 \leq j \leq h$) are given. Then there exists a randomized algorithm for finding a
        $X_{i,j}$-$Y_{i,j}$ vertex cut of size at most $t$ (if it exists) for all $1 \leq i \leq N$ and $1 \leq j \leq h$ in $\tilde{O}(t\TW{G}D + ht\TW{G})$ rounds whp.
\end{corollary}

\begin{corollary} \label{corol:multibroadcast}
\sloppy{
Let $\mathcal{H} = \{H_1, H_2, \dots, H_N\}$ be a near disjoint collection of $G$.
Assume that each subgraph $H_i \in \mathcal{H}$ contains $h$ source nodes 
$u_{i, 0}, u_{i, 1}, \dots, u_{i, h-1}$ (some of which might be the same node). 
They respectively broadcast messages $m_{i,0}, m_{i,1}, \dots, m_{i, h-1}$ of $O(\log n)$ bits 
in $H_i$. The output of the algorithm at node $v \in V(H_i)$ is the set of all source-message pairs 
$\{(u_{i, 0}, m_{i, 0}), (u_{i, 1}, m_{i, 1}), \dots, (u_{i, h-1}, m_{i, h - 1})\}$ in $H_i$. There exists a
randomized algorithm for solving this task within $\tilde{O}(\TW{G}D + h\TW{G})$ rounds whp.
}
\end{corollary}
We refer to the operations above as MVC($h, t$) and BCT($h$), respectively. 

\section{Fully Polynomial-Time Distributed Tree Decomposition}
\label{appendix:treewidth}

\subsection{Correctness of \textsc{Sep}}
\label{appendix:correctnesssep}

We prove the correctness of \textsc{Sep}. Let $S$ 
be any $(X, 1/2)$-balanced separator of size $\TW{G} + 1$ in $G$. Let $\Tcal^{\ast} = 
\bigcup_{1 \leq i \leq \hat{t}} \Tcal_i$. An ordered pair $(T_1, T_2) 
\in \Tcal^{\ast} \times \Tcal^{\ast}$ is called \emph{separated} if $T_1$ and $T_2$ do not intersect $S$ 
and belong to different connected components in $G - S$.



\begin{proposition} \label{prop:numberoftrees}
Assume that the algorithm executes step 4.
Then, $\forall i \in [1, \hat{t}], 3.9t \leq |\Tcal_i| \leq 12.1t$ holds. 
\end{proposition}

\begin{proof}
Let us denote $N = |\Tcal_i|$.
Since $R^{\ast}_{i-1}$ is not a $(X, 14399/14400)$-balanced separator, 
we get that $\mu(G_i) \geq 14399\mu(V(G))/14400$. Each tree $T \in \Tcal_i$ has a size at most $\mu(G)/(4t)$.
Thus the following inequality holds.
\begin{align*}
    N &\geq \frac{\mu(G_i)}{\frac{\mu(G)}{4t}} = \frac{4t\mu(G_i)}{\mu(G)} \geq \frac{4t \cdot 14399}{14400} \geq 3.9t.
\end{align*}
Let $\Tcal_i = \{T_1, T_2, \dots, T_N\}$.
Trees in $\Tcal_{i}$ are vertex disjoint except for their root vertices, and thus 
the collection of vertex subsets $\{\hat{V}_j\}_{1 \leq j \leq N}$, where $\hat{V}_j = V(T_j) \setminus R_i$, 
are disjoint.  Then we have $\mu(\hat{V}_j) = \mu(T_j) - 1 \geq \mu(G_i)/(12t) - 1$.
In addition, since the algorithm does not halt at step 1, we
obtain $1/\mu(G) \leq 1/200t^2 \leq 1/800$ (recall $t \geq \TW{G} + 1 \geq 2$). 
Putting everything together, we state the following inequality:
\begin{align*}
    N &\leq \frac{\mu(G_i)}{\frac{\mu(G)}{12t} - 1} \leq 
    \frac{\mu(G)}{\frac{\mu(G)}{12t} - 1} \leq \frac{12t}{1 - 1/\mu(G)} \leq (12t) \cdot \frac{800}{799}
    \leq 12.1t.
\end{align*}
\end{proof}

\begin{lemma} \label{lma:separatedpairs}
Assume that \textsc{Sep} executes step 4. There exists a subset $I \subset [1, \hat{t}]$ of cardinality
$\lceil t/300 \rceil$ such that for any $i \in I$, $|\Tcal_i|^2/20$ pairs in $\Tcal_i \times \Tcal_i$ 
are separated.
\end{lemma}

\begin{proof}
Since $|S| \leq t$ holds and $R_i$ for all $1 \leq i \leq \hat{t}$ are disjoint, 
there exists a subset $I \subseteq [1, \hat{t}]$ of cardinality $\hat{t} - t \geq \lceil t/300 \rceil$ such that
$R_i \cap S = \emptyset$ holds for any $i \in I$. We show that $I$ satisfies the condition 
of the lemma. Let $i$ be any index in $I$. By Proposition~\ref{prop:numberoftrees}, we have $\geq |\Tcal_i| \geq 3.9t$.
Since trees in $\Tcal_i$ are vertex disjoint except 
for roots, at least $|\Tcal_i| - t$ trees in $\Tcal_i$ do not intersect 
$S$. Letting $\Tcal'_i \subseteq \Tcal_i$ be the set of such trees, we have 
\begin{align*}
|\Tcal'_i| &\geq |\Tcal_i| - t \geq |\Tcal_i| - \frac{|\Tcal_i|}{3.9} \geq \frac{2.9|\Tcal_i|}{3.9}.
\end{align*}
Next, we bound the number of tree pairs in $\Tcal'_i \times \Tcal'_i$ which respectively 
belong to two distinct connected components in $G - S$. Let $H_1, H_2, \dots, H_N$ be 
the connected components of $G - S$,
$Y = \bigcup_{T \in \Tcal_i \setminus \Tcal'_i} V(T)$, and $H'_i$ be the subgraph of $H_i$ 
obtained by removing all the vertices in $Y$.
We denote by $x_k$ the number of trees in $\Tcal'_i$ which are contained in $H'_k$. 
Without loss of generality, we assume that $x_k \geq x_{k+1}$ holds 
for $1 \leq k \leq N-1$. 
Since $S$ is a $(X, 1/2)$-balanced separator, $\mu(H'_1) \leq \mu(H_1) \leq \mu(G)/2$ holds.
The number of vertices in each tree in $\Tcal'_i$ excluding its root is at least $\mu(G)/(12t) - 1$.
By the calculation similar with Proposition~\ref{prop:numberoftrees}, it follows that 
$H'_1$ contains at most $\mu(H'_1)/(\mu(G)/(12t) - 1) \leq 6.05t$ trees. 
That is, $t \geq x_1/6.05$ holds.
Let $\overline{H'_1} = H'_2 + H'_3 + \dots + H'_N$, and $\overline{x_1} = x_2 + x_3 + \dots, x_N$.
Since each tree in $\Tcal_i \setminus \Tcal'_i$ contains at most $\mu(G)/(4t)$
vertices and at most $t$ trees are contained in $\Tcal_i \setminus \Tcal'_i$, we have $|Y| \leq \mu(G)/4$.
By the fact of $\mu(H'_1) \leq \mu(G)/2$,
we obtain $\mu(\overline{H'_1}) \geq \mu(G)/2 - |Y| \geq \mu(G)/4$.
Trees in $\Tcal'_i$ have sizes at most $\mu(G)/(4t)$ and are vertex disjoint except for 
their roots, each tree in $\Tcal'_i$ can exclusively contains at most $\mu(G)/(4t) - 1$ vertices in 
$\overline{H_1}$.
Thus we can bound $\overline{x_1}$ as follows:
\begin{align*}
\overline{x_1} \geq \frac{\mu(\overline{H'_1})}{\mu(G)/(4t) - 1} \geq \frac{t\mu(G)}{\mu(G) - 4t} 
\geq t \geq \frac{x_1}{6.05}. 
\end{align*}
Putting all together, we have 
\begin{align*}
\frac{6.05|\Tcal'_i|}{7.05} &= \frac{6.05(x_1 + \overline{x_1})}{7.05} \geq \frac{6.05x_1 + x_1}{7.05} 
= x_1.
\end{align*}
Then there exists $j$ such that 
$|\Tcal'_i|/7.05 \leq \sum_{1 \leq k \leq j} x_k \leq 6.05|\Tcal'_i|/7.05$ holds. We define $\Tcal^+_i$
as the trees in $\Tcal'_i$ contained in $H_1 + H_2 + \dots + H_j$, and $\Tcal^-_i$ as those contained
in $H_{j+1} + H_{j+2} + \dots, + H_N$. Since any pair in $\Tcal^+_i \times \Tcal^-_i$ is separated,
$|\Tcal^+_i|(|\Tcal'_i| - \Tcal^+_i|)$ pairs in $\Tcal_i \times \Tcal_i$ are separated.
$\Tcal^+_i$ contains at least $|\Tcal'_i|/7.05$ trees, we conclude 
that $6.05|\Tcal'_i|^2/(7.05 \cdot 7.05) = (6.05 \cdot (2.9)^2) |\Tcal_i|^2 /(7.05 \cdot 3.9)^2 \geq |\Tcal_i|^2/20$ pairs are separated.
\end{proof}

In the following argument, let $I$ be the set of indices shown in Lemma~\ref{lma:separatedpairs}.
For simplicity of argument, we assume $I = [1, \lceil t/300 \rceil]$. Assume that the algorithm executes 
step 4. 
For $i \in I$, let $(T_{i,j,1}, T_{i,j,2})$ be the $j$-th pair chosen from $\Tcal_i \times \Tcal_i$ 
at step 4 ($1 \leq 1 \leq 2t$, $1 \leq j \leq 95$), and $\Pcal_i = \{(T_{i,j,1}, T_{i,j,2}) 
\mid j \in [1, 95]\}$, 
For a subset $\mathcal{X} \subseteq \Tcal^{\ast} \times \Tcal^{\ast}$, we define $V(\mathcal{X}) = \bigcup_{(T_1, T_2) \in \mathcal{X}} (V(T_1) \cup V(T_2))$. 

\begin{lemma} \label{lma:separateprobability}
With probability at least 4/5, there exists a subset $I' \subseteq I$ and a set of pairs 
$U = \{(T_{i,1}, T_{i,2})\}_{i \in I'}$ satisfying the following two conditions:
\begin{itemize}
    \item $(T_{i, 1}, T_{i, 2})$ is contained in $\Pcal_i$ and is separated.
    \item For any $\gamma : I' \to \{1, 2\}$, $\bigcup_{i \in I'} V(T_{i, \gamma(i)}) \geq \mu(G)/14400$.
\end{itemize}
\end{lemma}

\begin{proof}
We say that a tree $T \in \Tcal_i$ has \emph{a large intersection} with a vertex subset 
$V' \subseteq V(G)$ if $\mu(V' \cap V(T)) > \mu(G)/(24t)$ holds. Similarly, a pair of trees
$(T_1, T_2) \in \Pcal_i$ has a large intersection with $V'$ if either $T_1$ or $T_2$ has a large 
intersection with $V'$. We construct $U$ recursively as follows:
\begin{itemize}
\item Add an arbitrary separated pair in $\Pcal_1$ to $U$ as $(T_{1, 1}, T_{1, 2})$ if it exists.
\item Let $U_i = \{(T_{i', 1}, T_{i', 2}) \mid i' \leq i \}$. For $i \in [2, \lceil t/300 \rceil]$, 
add an arbitrary 
separated pair in $\Pcal_i$ not having a large intersection with $V(U_{i-1})$ to $U$ as 
$(T_{i, 1}, T_{i, 2})$ if it exists. 
\end{itemize}
The index $i$ is added to $I'$ if some pair is added as $(T_{i, 1}, T_{i, 2})$.
It is easy to see that this construction satisfies the first condition of the lemma. 
Hence we focus on the proof of the second condition. 

We define the indicator random variables $Q_i$ for $i \in I$ as $Q_i = 1$ if 
and only if $i \in I'$ (i.e., a pair $(T_{i,1}, T_{i, 2})$ in $\Pcal_i$ is added to $U$). 
Let $Q = \sum_{i' \in I} Q_{i'}$.  
Since $\mu(T_{i,j}) \leq \mu(G)/(4t)$ holds and $|I| \leq t/300$, we obtain the following 
inequality for any $i \leq |I| - 1$.
\begin{align*}
V(U_i) \leq \frac{(|I| - 1)\mu(G)}{2t} \leq \frac{\mu(G)}{600}.   
\end{align*}
Since trees in $\Tcal_i$ are vertex disjoint except for roots,
the number of trees in $\Tcal_i$ having 
a large intersection with $V(U_{i-1})$ is bounded as follows:
\begin{align*}
\frac{\frac{\mu_X(G)}{600}}{\frac{\mu_X(G)}{24t} - 1} \leq \frac{t}{22} \leq \frac{|\Tcal_i|}{80},
\end{align*}
where the last inequality comes from Proposition~\ref{prop:numberoftrees}.
Since an ordered pair in $\Tcal_i \times \Tcal_i$ has a large intersection with $V(U_{i-1})$
when at least one of the pair has a large intersection with $V(U_{i-1})$, 
at most $2 (|\Tcal_i| \cdot |\Tcal_i|/80)$ ordered pairs in $\Tcal_i \times \Tcal_i$
have a large intersection with $V(U_{i-1})$. Combined with Lemma~\ref{lma:separatedpairs}, 
there exist $(1/20 - 1/40) \geq |\Tcal|^2/40$ pairs
in $\Tcal_i \times \Tcal_i$ are separated and have no large intersection with $V(U_{i-1})$.
It implies $\Pr[Q_i = 0] \leq (1 - 1/40)^{95} \leq 1/10$, i.e., $\Pr[Q_i = 1] \geq 9/10$.
Then we have $E[Q] \geq 9t/3000$ and thus $E[(t/300 - Q)] < t/3000$. 
By Markov inequality, $\Pr[(t/300 - Q) >  t/600] \leq (t/3000) / (t/600) = 1/5$. 
Thus we have $\Pr[Q > t/600] \geq 4/5$.

Let $U'_i(\gamma) = \bigcup_{i' \in I' \wedge i' \leq i} V(T_{i', \gamma(i')})$.
Since $V(U'_i(\gamma)) \subseteq V(U_i)$ holds
for any $i \geq 1$ satisfying $i \in I'$, both $T_{i,1}$ and $T_{i,2}$ do not have a large 
intersection with $V(U'_{i-1})$ (where $V(U'_{0})$ is defined as the empty set). 
Thus $\mu(V(U'_{i}(\gamma))) \geq \mu(G)/(24t) + \mu(V(U'_{i-1}(\gamma))$ 
if $i \in I'$ and $\mu(V(U'_{i}(\gamma))) \geq \mu(V(U'_{i-1}(\gamma))$ otherwise. 
Since $|I'| > t/600$ holds, we obtain $U'_{|I'|}(\gamma) \geq \mu(G)/14400$. The lemma is proved. 
\end{proof}

\begin{lemma} \label{lma:balancedseparator}
\textsc{Sep} outputs a $(X, 14399/14400)$-balanced separator of size at most $400(\TW{G} + 1)^2$ with high probability.
\end{lemma}

\begin{proof}
Since the value $t$ such that $\TW{G} + 1 \leq t < 2(\TW{G} + 1)$,
it suffices to show that the algorithm outputs a $(X, 14399/14400)$-balanced separator of size at most 
$200t^2$ if $t \geq \TW{G} + 1$ holds. The proof is obvious if the algorithm outputs the separator 
at step 1. Consider the case that the algorithm outputs
a separator at step 3 in the $i$-th iteration. By Proposition~\ref{prop:numberoftrees}, 
$|R_i| \leq |\Tcal_{i'}| \leq 12.1t$ holds for any $i' \leq i$. Thus we obtain $|R^{\ast}_i| \leq 
i \cdot 12.1t \leq 301t/300 \cdot 
12.1t \leq 13t^2$. Next, consider the case that the algorithm outputs a separator at step 5. 
Let $U = \{(T_{i,1}, T_{i,2})\}_{i \in I'}$ be the set of pairs by Lemma~\ref{lma:separateprobability}.
Since each pair $(T_{i,1}, T_{i,2})$ is separated, the minimum $V(T_{i,1})$-$V(T_{i,2})$ vertex cut
has a size at most $t$. Hence it is contained in $Z$. Then either $T_{i, 1}$ or $T_{i,2}$ does not 
belong to the largest connected component in $G - Z$.
We define $\gamma(i)$ as the value such that $T_{i,\gamma(i)}$ does not belong to the largest one.
Then Lemma~\ref{lma:separateprobability} implies that a vertex subset of size $\mu(G)/14400$ is not
contained in the largest connected component in $G - Z$, i.e., $Z$ is a $(X, 14399/14400)$-balanced
vertex separator of $G$. The success probability for one trial is at least $4/5$ by 
Lemma~\ref{lma:separateprobability}.
The total success probability is at least $1 - (1/5)^{5\log n} \geq 1 - 1/n^{10}$.
The exponent 10 of the failure probability $1/n^{10}$ can be an arbitrarily large constant by tuning 
the constant parameters of the algorithm.
\end{proof}

\subsection{Distributed Implementation of \textsc{Sep}}
\label{appendix:distsep}

Now we are ready to prove Lemma~\ref{lma:balancedseparatorconstruction}.

\begin{rlemma}{lma:balancedseparatorconstruction}
    Let $G$ be any undirected graph, and $X \subseteq V(G)$ be any vertex subset. There exists 
    a randomized CONGEST algorithm which outputs a $(X, 14399/14400)$-balanced 
    separator of size at most $400(\TW{G}+1)^2$ for $G$ within $\tilde{O}(\TW{G}^2D + \TW{G}^{3})$ 
    rounds whp. 
\end{rlemma}

\begin{proof}
The correctness of the output is derived from Lemma~\ref{lma:balancedseparator}.
Assume $\TW{G} + 1 \leq t$. The cost of estimating $t$ is an extra $O(\log n)$ multiplicative 
factor in the running time, and thus we omit it. Each step of \textsc{Sep} is implemented as follows. 
\begin{itemize}
\item (Step 1) 
This is trivially implemented by counting the number of vertices in $X$ using PA. 
\item (Step 2)
Since the trees in $\Tcal \cup \Tcal_i$ are disjoint
except for their roots, one can
identify each tree $T \in \Tcal \cup \Tcal_i$ by the ordered pair of the root ID and the maximum of 
its children's IDs, which we regard as the ID of $T$. Each node $u$ in
$T \in \Tcal_i$ manages the \emph{profile} of $T$ consisting of the value $i$, the ID of $T$, and 
the total size $\mu(T)$. If $T$ belongs to $\Tcal$, $i = 0$ holds.
When a tree $T$ is added to $\Tcal_i$, the root node of $T$ is guaranteed to know the profile of $T$.
The addition of $T$ to $\Tcal_i$ is done by propagating the profile of $T$ to all nodes in $T$.
Since $\Tcal \cup \Tcal_i$ are the collection of the subtrees which can intersect only at their roots, it is 
near-disjoint, Hence one can do this propagation using BCT(1) for $\Tcal \cup \Tcal_i$.

In the implementation of step 2, the algorithm first constructs a rooted spanning tree $T^{\ast}$ by
RST for $G$, and sets $\Tcal = \{T^{\ast}\}$ by the profile propagation above. 
The nodes of each tree $T' \in \Tcal$ execute the \textsc{Split} procedure.
Each node $v \in V(T')$ computes the size $\mu(T'(v))$ of 
the subtree rooted at $v$ (STA), and finds $\mu(T'(u))$ for 
$u \in \CH(T', v)$ via communication with its children (SNC). One can choose as a center, $c$, any node 
$v \in V(T')$ satisfying (1) $\mu(T'(u)) \leq \mu(T')/2$ for any $u \in \CH(T', v)$ and (2) 
$\mu(T'(v)) \geq \mu(T')/2$. If two or more nodes satisfy the conditions, an arbitrary one of them 
is chosen by leader election (SLE). The algorithm replaces the root of $T'$ by the chosen $c$. 
The rooted tree after the replacement is referred to as $T''$. The detection of the new parent
of each node in $T''$ is implemented by STA: Consider the input of value one for $c$ and zero 
otherwise, and executing STA for $T'$ with respect to this input. After running STA, each node in 
$T'$ exchanges the output value with its neighbors (SNC). Each node $u \in T'$ identifies the child 
$v \in \CH(T', u)$ with output value one as the new parent in $T'_c$. If no child has output 
value one, the parent of $u$ in $T'$ is also recognized as the parent in $T''$.

Since $c$ knows the size of $T''(u)$ for all $u \in \CH(T'' c)$, it can handle the task of splitting 
$T''$ locally, and then $c$ puts each split tree into $\Tcal$ or $\Tcal_i$ by propagating its profile.
The process above consists only of the subgraph operations for 
$\Tcal \cup \Tcal_i$. Since \textsc{Split} is invoked $O(\log t)$ times, Step 2 is realized 
by $\tilde{O}(1)$ invocations of
PA, SNC, SLE, RTA, STA, and BCT(1).
\item (Step 3) The algorithm first detects all connected components in $G_i - R_i$ using CCD, 
and then computes the size of each detected component using PA.  
If some component has a size larger than $14399\mu(G)/14400$, it becomes $G_{i+1}$. 
\item (Step 4) The algorithm elects a leader node $r \in V(G)$. Each root node of the trees in $\Tcal_i$ for all 
$1 \leq i \leq \hat{t} $ sends
its profile to $r$. By Proposition~\ref{prop:numberoftrees}, we have the bound $|T_i| \leq 12.1t$. Hence this task 
is processed by multi-source broadcast BCT($12.1t\hat{t}$) presented in 
Corollary~\ref{corol:multibroadcast}. Then
the root node $r$ locally samples 95 pairs of trees from $\Tcal_i \times \Tcal_i$ using the received profiles,
and broadcasts the profiles of trees in sampled pairs to all of the nodes in $V(G)$ using BCT$(95\hat{t})$. 
The nodes in each sampled tree set up the input for the task of MVC. Finally, the algorithm executes 
MVC$(95\hat{t}, t+1)$.  
\end{itemize}
By Lemma~\ref{lma:SAapplication}, the running times of steps 1-3 are all bounded by 
$\tilde{O}(\TW{G}D)$ rounds. Since steps 2-3 are repeated $O(t)$ times, the total 
running time spent for those steps is $\tilde{O}(t\TW{G}D)$ rounds. Step 4 consists of invocations of
BCT($12.1t\hat{t}$), BCT($95\hat{t}$), and MVC$(95\hat{t}, t+1)$.
By Corollaries~\ref{corol:multivertexcut} and \ref{corol:multibroadcast}, this step 
takes $\tilde{O}(t\TW{G}D + t^2\TW{G})$ rounds. The algorithm terminates when $t \geq \TW{G} + 1$ holds. Since
$t$ is estimated via doubling, it never becomes greater than $2(\TW{G} + 1)$. Hence throughout the
execution of the algorithm $t = O(\TW{G})$ holds. That is, the total running time of the algorithm is $\tilde{O}(\TW{G}^2D + 
\TW{G}^3)$. 
\end{proof}

\subsection{Distributed Tree Decomposition Based on Balanced Separators}
\label{appendix:DistTreeDecomposition}

We first show that the output of the proposed algorithm satisfies the definition
of tree decomposition. We first present an auxiliary proposition.

\begin{proposition} \label{prop:gateway}
For any $x \in A_\ell(T)$ such that $\ell > 0$, $G'_{x}$ is a connected component 
of $G - B_{\Par{x}}$.
\end{proposition}

\begin{proof}
The proof is by induction on the length of $x$. (Base) If $|x| = 1$, we have  
$B_{\Par{x}} = S_{\psi}$  and $G_{\Par{x}} = G$. Since $G'_{x}$ is a connected component of $G_{\Par{x}} - B_{\Par{x}}$, the proposition obviously holds. 
(Induction step) Assume that the proposition holds for any $x' \in A_\ell(T)$ and consider $x \in A_{\ell + 1}(T)$. 
By the definition of bags $B_x$, we obtain the following equalities:
\begin{align*}
B_{\Par{x}} &= V(G_{\Par{x}}) \cap \bigcup_{x' \sqsubseteq \Par{x}} S_{x'}, \\
B_{\Par{\Par{x}}} &= V(G_{\Par{\Par{x}}}) \cap \bigcup_{x' \sqsubseteq \Par{\Par{x}}} S_{x'}.
\end{align*}
It is obvious that $V(G_{\Par{x}}) \subseteq V(G_{\Par{\Par{x}}})$ holds. Hence we obtain
\begin{align*}
V(G_{\Par{x}}) \cap B_{\Par{\Par{x}}} &= V(G_{\Par{x}}) \cap \bigcup_{x' \sqsubseteq \Par{\Par{x}}} S_{x'}, 
\end{align*}
and thus the following equality holds
\begin{align*}
B_{\Par{x}} &= \left(V(G_{\Par{x}}) \cap \bigcup_{x' \sqsubseteq \Par{\Par{x}}} S_{x'}\right) \cup \left(V(G_{\Par{x}}) \cap S_{\Par{x}} \right) \\
& = \left( V(G_{\Par{x}}) \cap B_{\Par{\Par{x}}} \right) \cap \left(V(G_{\Par{x}}) \cap S_{\Par{x}} \right) \\
& = V(G_{\Par{x}}) \cap (B_{\Par{\Par{x}}} \cup S_{\Par{x}}).
\end{align*}
The graph $G'_{x}$ is a connected component of $G_{\Par{x}} - B_{\Par{x}}$, i.e., a connected component of
$G_{\Par{x}} - (V(G_{\Par{x}}) \cap (B_{\Par{\Par{x}}} \cup S_{\Par{x}})) = G_{\Par{x}} - (B_{\Par{\Par{x}}} \cup S_{\Par{x}})$ 
by the equality above. In addition, removing all vertices in $B_{\Par{\Par{x}}}$ from $G_{\Par{x}}$ 
induces the graph $G'_{\Par{x}}$. Thus we can conclude that $G'_{x}$ is a connected component of $G'_{\Par{x}} - S_{\Par{x}}$. 
By the induction hypothesis, $B_{\Par{\Par{x}}}$ separates $G'_{\Par{x}}$ from $G$. 
That is, $B_{\Par{x}} = V(G_{\Par{x}}) \cap (B_{\Par{\Par{x}}} \cup S_{\Par{x}})$ separates $G'_{x}$ from $G$. 
The proposition holds.
\end{proof}

The following lemma guarantees the correctness of our tree decomposition. 

\begin{lemma} \label{lma:correctnessTreeDecomposition}
The decomposition $\Phi = (T, \{B_x\}_{x \in V(T)})$ is a tree decomposition
of width at most $O(\TW{G}^2 \log n)$ for any graph $G=(V, E)$. The depth 
of $T$ is $O(\log n)$.

\end{lemma}
\begin{proof}
It is obvious that any vertex is contained in at least one bag, and thus condition (a) of the definition of 
tree decomposition (see section~\ref{subsec:treedecomposition}) is satisfied. 
Let $u \in V(G)$ be any node, and $X \subseteq V(T)$ be the set of nodes such that their
corresponding bags contain $u$.
To prove condition (c), it suffices to show that the subgraph $T'$ of $T$ induced by $X$ is connected.
Consider any two nodes $x$ and $y$ in $X$. 

(Case 1)
If $x \sqsubseteq y$ holds, for any $x \sqsubseteq z \sqsubseteq y$, we obtain $V(G_{y}) \subseteq V(G_{z}) \subseteq V(G_{x})$. Then we have $V(G_{z})$ contains $u$. 
Since $B_x = V(G_x) \cap (\bigcup_{x' \sqsubseteq x} S_{x'})$ holds, we have
$u \in \bigcup_{x' \sqsubseteq x} S_{x'} \subseteq \bigcup_{z' \sqsubseteq z} S_{z'}$ holds.
Consequently, $u \in B_{z} = V(G_{z} \cap (\bigcup_{z' \sqsubseteq z} S_{z'})$ holds, which implies that $x$ 
and $y$ is connected in $T'$.

(Case 2) If $x \parallel y$ holds, let $z$ be the 
longest common prefix of $x$ and $y$, and $G_{z \bullet i}$ and $G_{z \bullet j}$ be the graph containing 
$G_{x}$ and $G_{y}$ as subgraphs respectively ($i \neq j$). 
Since $u$ is contained both $G_{z \bullet i}$ and 
$G_{z \bullet j}$, we have $u \in V(G_{z \bullet i}) \cap V(G_{z \bullet j}) \subseteq B_{z}$. 
Thus we can reduce this case to case 1, i.e., the connectivity between $z$ and $x$ satisfying  
$z \sqsubseteq x$ is proved and that between $B_z$ and $B_y$ satisfying $z \sqsubseteq y$ is also proved. 
It follows that condition (c) is satisfied.

By conditions (a) and (c), the canonical string $\St_{\Phi}(u)$ is well-defined for any $u \in V(G)$. 
Let $e = (u, v)$ be any edge in $V(G)$. We first show $\St_\Phi(u) \not\parallel \St_\Phi(v)$. 
Suppose for contradiction that $\St_\Phi(u) \parallel \St_\Phi(v)$ holds.
Letting $z$ be the longest common prefix of $\St_\Phi(u)$ and $\St_\Phi(v)$,
$B_{z}$ contains neither $u$ nor $v$ by the definition of canonical strings.
If $G_{z \bullet i}$ for some $i$ contains both $u$ and $v$,
it contradicts that $z$ is the longest. Otherwise, $B_z$ separates $u$ and $v$ 
into two different connected components in $G_x - B_z$, which also contradicts 
the existence of the edge $(u, v)$ in $G$. 
Without loss of generality, we assume $\St_\Phi(u) \sqsubseteq \St_\Phi(v)$.
Then it suffices to show $u \in B_{\St_\Phi(v)}$. Let us denote $x = \St_\Phi(v)$ for brevity. 
By the definition of canonical strings, 
$v \not \in B_{x'}$ for any $x' \sqsubset x$. That is, $v \not\in B_{\Par{x}}$ and $v \in B_{x} \subseteq V(G_{x})$ holds. Since $V(G_{x}) \subseteq V(G'_{x}) \cup B_{\Par{x}}$ holds,
it implies $v \in V(G'_x)$. By Proposition~\ref{prop:gateway}, $B_x$ separates $G'_x$ from $G$. Due to the existence of the edge $(u, v)$, it must be that $u$ is
contained in $B_x$, i.e., condition (b) is satisfied. 

By Lemma~\ref{lma:balancedseparatorconstruction}, each separator has a size of $O(\TW{G}^2)$. The depth $O(\log n)$ of $\Tcal$ is derived from the $14399/14400$-balanced separation property.
Thus we obtain that the width of the decomposition is $O(\TW{G}^2\log n)$. 
\end{proof}

We prove that our tree decomposition algorithm can be efficiently implemented in the CONGEST model. 

\begin{rtheorem}{thm:treedecomposition}
For a given graph, $G=(V, E)$, there exists an algorithm in the CONGEST model which constructs a tree decomposition,
$\Phi = (T, \{B_x\}_{x \in V(T)})$, of width $O(\TW{G}^2 \log n)$ whp. 
The depth of $T$ is $O(\log n)$ and the running time of the algorithm 
is $\tilde{O}(\TW{G}^2D + \TW{G}^{3})$ rounds.
\end{rtheorem}

\begin{proof}
The distributed implementation of the algorithm above is relatively straightforward. By Proposition~\ref{prop:gateway}, 
one can execute the algorithm of Lemma~\ref{lma:balancedseparatorconstruction} for computing
the balanced separators of all subgraphs in $\{G'_x \mid x \in A_\ell(T)\}$.
After computing $S'_{x}$ for all $x \in A_\ell(T)$, each node in $G_x$ checks if it belongs to 
the bag $B_x = V(G_x) \cap (\bigcup_{x' \sqsubseteq x} S'_{x'})$. Note that each node in $S'_{x'}$ knows
that it belongs to $S'_{x'}$ as well as the string $x'$. Hence the set of nodes in $B_x$ are identified
locally. The collection of subgraphs $\Gcal_{\ell + 1}$ is 
recognized by the connected component detection for subgraphs $\{G_x - B_x \mid x \in A_\ell \}$ (CCD). 
The output of a node $v$ in $S'_x$ is $x$. 
The correctness of the output follows from Lemma~\ref{lma:correctnessTreeDecomposition}.
The running time of the algorithm is dominated by that for computing balanced
separators (i.e., Lemma~\ref{lma:balancedseparatorconstruction}), i.e. $\tilde{O}(\TW{G}^2D + \TW{G}^3)$ rounds.
\end{proof}

\section{Distributed Distance Labeling in Low-Treewidth Graphs}
\label{appendix:distancelabeling}

\begin{rlemma}{lma:correctnessDL}
For any $u, v \in V(G)$, $\Ddec(\Dlabel_{G}(u), \Dlabel_{G}(v)) = d_G(u, v)$ holds.
\end{rlemma}

\begin{proof}
Let us denote $d = \Ddec(\Dlabel_G(u), \Dlabel_G(v))$ for brevity. 
(Case 1: $\St(u) \nparallel \St(v)$) Without loss of generality, we assume that
$\St(v) \sqsubseteq \St(u)$. Since $B^{\uparrow}_{\Phi}(v) \subseteq B^{\uparrow}_{\Phi}(u)$ 
holds, $\Dlabel_G(u)$
and $\Dlabel_G(v)$ respectively contain $(u, v, d_G(u, v))$ and $(v, v, d_G(v, v))$, and thus
$d = d_G(u, v)$ holds.
(Case 2: $\St(u) \parallel \St(v)$). Let $z$ be 
the longest common prefix of $\St(u)$ and $\St(v)$. Since $\St(u) \parallel \St(v)$ holds, none
of $u$ and $v$ belong to $B_{z}$. Then $B_{z}$ separates $G$ into 
several connected components such that $v$ and $u$ belong to different ones. Thus any shortest 
path $p$ from $u$ to $v$ necessarily 
intersects $B_{z}$. Let $s$ be any vertex in $B_{z} \cap V(p)$. Since $B^{\uparrow}_{\Phi}(u)$ and $B^{\uparrow}_{\Phi}(v)$ 
contains $B_{z}$ as a subset, $B^{\uparrow}_{\Phi}(u)$ and $B^{\uparrow}_{\Phi}(v)$ respectively contain
$(u, s, d_G(u, s))$ and $(s, v, d_G(s, v))$. It implies $d = d_G(u, v)$.
The lemma is proved.
\end{proof} 

\begin{rlemma}{lma:distH}
For any $u, v \in V(H_x)$, $d_{H_x}(u, v) = d_{G_x}(u, v)$ holds.
\end{rlemma}

\begin{proof}
Let $p = v_0, v_1, \dots, v_{\ell - 1}$ be a shortest path from $u$ to $v$ (i.e., $u = v_0$ and
$v = v_{\ell - 1}$ in $G_{x}$, and $\{v_{i_0}, v_{i_1}, \dots, v_{i_{j-1}}\}$ $(i_h < i_{h+1}$ for all $0 \leq h < j - 1$, $i_0 = 0$, and $i_{j-1} = \ell -1$) be the set of vertices in $V(p) \cap V_{H_x}$. We denote by $p(h)$ the subpath of $p$ from $v_{i_h}$ 
to $v_{i_{h+1}}$.
Since $p(h)$ contains no intermediate vertex in $H_{x}$, each $p(h)$ is contained in
a graph $G'_{x \bullet i}$ for some $i \in \CH(T, x)$, or consists of a single edge $(v_{i_h}, v_{i_{h+1}})$. It implies that $H_{x}$ has the edge 
$(v_{i_h}, v_{i_{h+1}})$ of costs $c_{H_x}(v_{i_h}, v_{i_{h+1}}) = d_{G'_x}(v_{i_h}, v_{i_{h+1}})$. Then $H_x$ contains a path $v_{i_0}, v_{i_1}, \dots, v_{i_{j-1}}$ of length 
$\sum_{0 \leq h \leq j - 1} c_{H_x}(v_{i_h}, v_{i_{h+1}}) = \sum_{0 \leq h \leq j - 1} d_{G_x}(v_{i_h}, v_{i_{h+1}}) = d_{G_x}(u, v)$. That is, we have $d_{H_x}(u, v) \leq d_{G_x}(u, v)$. We also obtain $d_{H_x}(u, v) \geq d_{G_x}(u, v)$ by the converse argument.
\end{proof}

\begin{rlemma}{lma:distUpdate}
Let $u$ and $v$ be any vertex in $V(G_{x \bullet i}) \cup B_x$ for some $i \in \CHT(x)$. Then the following equality holds.
\begin{align*}
d_{G_x}(u, v) &= 
\min \{d_{G'_{x \bullet i}}(u, v), \\ 
& \quad \quad \min_{s, s' \in V(H_x)} (d_{G_{x \bullet i}}(u, s) + d_{H_x}(s, s')  + d_{G_{x \bullet i}}(s', v))\}. 
\end{align*}
\end{rlemma}

\begin{proof}
Let $p$ be a shortest path from $u$ to $v$. If $p$ does not contain any vertex in $B_{x}$, $p$ is a path in $G'_{x \bullet i}$. 
Thus we obtain $d_{G_x}(u, v) = d_{G'_{x \bullet i}}(u, v) = d_{G_{x \bullet i}}(u, v)$. Otherwise, $p$ contains a vertex in 
$B_{x}$. Let $s$ and $s'$ be the first and last vertices in $B_{x} \cap V(p)$ in $p$. Then the length 
of $p$ (i.e. $d_{G_x}(u, v)$) is equal to $d_{G_{x \bullet i}}(u, s) + d_{G_x}(s, s') + d_{G_{x \bullet i}}(s', v)$. 
By Lemma~\ref{lma:distH}, $d_{G_x}(s, s') = d_{H_x}(s, s')$ holds. Thus we have
$d_{G_x}(u, v) = d_{G_{x \bullet i}}(u, s) + d_{H_x}(s, s') + d_{G_{x \bullet i}}(s', v)$. 
\end{proof}

\begin{rtheorem}{thm:dlabel}
Let $G = (V(G), E(G))$ be any directed graph with edge cost function $c:E(G) \to \mathbb{N}$.
Then there exists a randomized CONGEST algorithm that solves $\textsf{DL}$ in 
$\tilde{O}(\TW{G}^2D + \TW{G}^5)$ rounds with probability at least $1 - 1/n^{9}$. The label size of each node 
is $O(\TW{G}^2\log^2n)$ bits. 
\end{rtheorem}

\begin{proof}
The correctness of the algorithm follows from Lemmas~\ref{lma:distH} and \ref{lma:distUpdate}.
Except for the recursive calls, the running time of the algorithm above is dominated by the broadcast operations in steps 1 and 3, which are implemented by BCT($h$) (according to Corollary~\ref{corol:multibroadcast}). 
Since we utilize the algorithm of Theorem~\ref{thm:treedecomposition}, the width of the given tree decomposition
is $\tilde{O}(\TW{G}^2)$, and thus the amount of information broadcast in steps 1 and 3 is 
$\tilde{O}(\TW{G}^4)$ bits. The running time for the task
is $\tilde{O}(\TW{G} D + \TW{G}^5)$. Since the recursion depth is bounded by the height of $T$, i.e., $O(\log n)$, 
the running time is $\tilde{O}(\TW{G}D + \TW{G}^5)$ rounds.
Combining this with the running time of our tree decomposition algorithm ($\tilde{O}(\TW{G}^2D + \TW{G}^{3})$ rounds), 
the total running time is bounded by $\tilde{O}(\TW{G}^2D + \TW{G}^5)$ rounds.
\end{proof}

\section{Finding Stateful Walks}
\label{appendix:statefulwalk}

\subsection{Proof of Lemma~\ref{lma:shorteststatefulwalk}}
\begin{rlemma}{lma:shorteststatefulwalk}
Let $G = (V(G), E(G), \gamma_G, f)$ be any labeled graph, edge-cost function 
$c:E(G) \to \mathbb{N}$, and $C \subseteq W_G$ be a stateful walk constraint with associated triple 
$(Q, \Mem, \boldsymbol{\Trans})$. There exists a walk $w$ of weighted length $x$ from $s$ to $t$ with state $q$ 
($s, t \in V(G)$, $q \in Q \setminus \{\perp\}$) in $C$ 
if and only if there exists a walk $w'$ of weight $x$ from $(s, \Auxinit)$ to $(t, q)$ in $G_C$. 
\end{rlemma}

\begin{proof}
We only prove the direction $\Rightarrow$. The proof for direction $\Leftarrow$ is easily obtained by inverting the argument below. 
Let $w = e_1, e_2, \dots, e_{x}$. By conditions 1 and 3 of stateful walk constraints,
any prefix $w_i$ of $w$ up to $e_i$ is also contained in $C$. We denote $u_0 = s$,
$u_i = \gamma_G(e_i)[1]$ for $i \in [1, x]$, $q_0 = \Auxinit$, and $q_i = \Mem(w_i)$ for 
$i \in [1, x]$. Condition 2 of the stateful walk constraints implies $\Trans_{e}(q_i) = q_{i+1}$. Hence we have an edge $e'_i = ((u_{i-1}, q_{i-1}), (u_i, q_i))$ for each 
$i \in [1, x]$
in $E(G_C)$. That is, $G_C$ contains a walk $w' = e'_1, e'_2, \dots e'_x$  
from $(s, \Auxinit)$ to $(t, q)$. Since each $e'_i$ has the same weight as $e_i$, the weighted length of $w'$ 
is equal to that of $w$, i.e., $x$.
\end{proof}

\section{Exact Bipartite Maximum Matching in Low-Treewidth Graphs}
\label{appendix:matching}

\begin{rtheorem}{thm:matching}
There exists a randomized CONGEST algorithm that computes the maximum matching for any bipartite graph $G$ 
in $\tilde{O}(\TW{G}^{4}D + \TW{G}^7)$ rounds whp.
\end{rtheorem}
\begin{proof}
The algorithm first finds a $O(1)$-balanced vertex separator $S$ of size $O(\TW{G}^2)$, and constructs the
maximum matching for each connected component of $G - S$. If the component is sufficiently small (i.e., $O(\TW{G}^2)$ vertices), the algorithm computes the maximum matching in the centralized fashion (as the algorithm of
Theorem~\ref{thm:dlabel}). Let $S = 
\{s_1, s_2, \dots, s_k\}$, and $S_i = \{s_i, s_{i+1}, \dots, s_k\}$. The algorithm sequentially adds 
vertices in $S$ to $G - S$ and updates the matching. Assume that the maximum matching of 
$G - S_i$ has been computed. By Proposition~\ref{prop:maximummatching}, it suffices to find an augmenting 
path from $s_{i}$ in $G - S_{i+1}$ to obtain the maximum matching of $G - S_{i+1}$, which can be done by
solving $\mathsf{CDL}(\Ccolor{2})$. If an augmenting path is found, the matching is updated locally. Otherwise, 
the maximum matching of $G - S_{i}$ is also the maximum matching of $G - S_{i+1}$.

The recursive construction for each connected component in $G- S$ is performed in parallel. 
Since the algorithm of Lemma~\ref{lma:balancedseparatorconstruction} supports the parallel 
computation of balanced separators for a collection of disjoint subgraphs, one can compute the balanced
separators for each connected component of $G - S$ in $\tilde{O}(\TW{G}^4D + \TW{G}^7)$ rounds.
The construction of $\mathsf{CDL}(\Ccolor{2})$ for each connected component is done by applying the 
algorithm of Theorem~\ref{thm:dlabel} for the entire graph $G$ while assigning all edges incident to 
a vertex in $S$ with cost $\infty$. The construction of augmenting paths in each connected component of
$G-S$ is also processed in parallel, because the algorithm of Corollary~\ref{corol:pathconstruction}
is implemented by the communication primitives of Appendix.~\ref{appendix:subgraphoperation} and thus it can
be executed in all connected components simultaneously. Consequently, the running time of the algorithm at each
recursion level is dominated by that for $O(\TW{G}^2)$ times of augmenting path findings, which takes 
$\tilde{O}(\TW{G}^4D + \TW{G}^7)$ rounds. The depth of the recursion is bounded by $O(\log n)$.
\end{proof}

\section{Computing Girth}
\label{appendix:girth}

As mentioned in Section~\ref{sec:girth}, we focus on the computation of the girth for undirected and weighted graphs. We first present the proof of Lemma~\ref{lma:count-one-girth}.
\begin{rlemma}{lma:count-one-girth}
Any shortest exact count-1 walk $w$ starting and terminating at the same vertex $v$ contains a simple cycle, and thus the weighted length of $w$ is at least $g$. 
\end{rlemma}

\begin{proof}
It suffices to show that if $w$ is not simple, $w$ contains a simple cycle (this obviously implies that 
$g \leq |w|$). Let $w = u_0, e_1, u_1, e_2, u_2, \dots, e_{\ell}, u_{\ell}$ be a shortest exact count-1 walk from $v$ to $v$ (i.e., $u_0 = u_\ell = v$), and assume $e_k$ has label one. Assume $w$ is not simple, and let $w' = u_0, e_1, u_1, e_2, u_2, \dots, e_j, u_j$ be the shortest non-simple prefix of $w$, where $u_i = u_j$ holds for some $0 \leq i < j$. The only possible case that $w'$ does not contain a simple cycle is
the situation of $e_{i+1} = e_j$. Suppose for contradiction that this situation applies. Then $e_{i+1}$ (= $e_j$) cannot have label one because $w$ contains exactly one label-one edge. Hence 
$i, j < k$ or $i, j \geq k$ holds. However, we can obtain a shorter exact count-1 walk from $v$ to $v$ by skipping the subwalk from $u_i$ to $u_j$, which contradicts that $w$ is the shortest exact count-1 walk from $v$ to $v$.
\end{proof}

Let $C = e_1, e_2, \dots, e_g$ be any shortest cycle in $G$, and assume that exactly one edge 
$e_i \in E(C)$ satisfies $f(e_i) = 1$. The algorithm for $\Ccount{1}$-distance labeling allows each node 
$u$ to compute the weighted length of the shortest count-1 walk from $u$ to $u$ (this can be computed 
locally from the output $\Slabel_{G,\Ccount{1}}(u)$). Letting $g(u)$ be the weighted length computed by 
a node $u \in V(G)$, the output value $g(v)$ of $v \in V(C)$ is equal to $g$. By Lemma~\ref{lma:count-one-girth}, one can compute the girth $g = \min_{u \in V(G)} g(u)$ by standard 
aggregation over all nodes. The remaining issue is how to find a function $f$ satisfying the above 
assumption. This can be solved by a probabilistic label assignment.
Let $E_C(G)$ be the set of edges in $G$ which belong to at least one shortest cycle in $G$,
and $c$ be the value of power 2, satisfying $|E_C(G)| \leq c < 2|E_C(G)|$. The value $c$ 
is obtained by a standard doubling estimation technique. We describe below the entire structure of our algorithm.
\begin{itemize}
    \item For each $\hat{c} = 1, 2, 4, \dots, 2^{\lceil \log n^2 \rceil + 1}$, repeat the following trial (steps 1 and 2) $O(\log n)$ times.
    \begin{enumerate}
    \item With probability $1/(3\hat{c})$, assign each edge with label one.
    \item Construct $\Slabel_{G, \Ccount{1}}(v)$ at each node $v \in V(G)$, and each node $v$ computes the length of the shortest exact count-1 walk from $v$ to $v$.
    \end{enumerate}
    \item The output of the algorithm is the minimum of all computed values over all nodes and trials.
\end{itemize}
To prove the theorem below it is sufficient to prove the correctness of the algorithm. 
\begin{rtheorem}{thm:girth}
There exists a randomized CONGEST algorithm that given a directed and weighted
graph $G$, computes its girth, $g$, in $\tilde{O}(\TW{G}^{2}D + \TW{G}^5)$ rounds whp.
\end{rtheorem}

\begin{proof}
By Lemma~\ref{lma:count-one-girth}, the output of the algorithm is obviously lower bounded by $g$. Thus, it suffices to show that the value $g$ is outputted at a node in some trial with high probability. 
Let $X_{\hat{c}}$ be the event that exactly one edge in $E_C(G)$ has label one at trial $\hat{c}$. 
At the trial $\hat{c} = c$, we have the following bound.
\begin{align*}
\Pr[X = \mathrm{true}] 
&\geq |E_C(G)| \cdot \frac{1}{3c} \left(1 -\frac{1}{3c}\right)^{|E_C(G)|} \\
&\geq \frac{1}{18}.
\end{align*} 
This implies that at the trial $\hat{c} = c$ some node in a shortest cycle computes the value $g$ with 
a constant probability. Repeating the trials $O(\log n)$ times sufficiently amplifies the success probability 
of the algorithm. 
\end{proof}


\end{document}